\documentclass[a4paper, 11pt]{amsart}
\usepackage[top=100pt,bottom=60pt,left=96pt,right=92pt]{geometry}
\usepackage{ytableau}
\usepackage{tikz}
\usepackage{setspace}
\usepackage{amsthm}
\usepackage{graphicx}

\newcommand{\Z}{\mathbb{Z}}
\newcommand{\ux}{\underline{x}}
\newcommand{\uy}{\underline{y}}
\newcommand{\uz}{\underline{z}}

\usepackage{makecell}

\newcommand{\mS}{\mathbb{S}}
\usepackage{enumerate}

\usepackage{hyperref}

\newcommand{\uv}{\underline{v}}
\newcommand{\uu}{\underline{u}}
\newcommand{\uU}{\underline{U}}

\newcommand{\up}{\underline{\partial}}

\usepackage{tikz-cd}
\usepackage{theoremref}
\usepackage{epigraph}
\usepackage{dynkin-diagrams}
\newcommand{\Alg}{\operatorname{Alg}}

\usepackage{dynkin-diagrams}
\usepackage{framed}
\definecolor{shadecolor}{gray}{0.9}
\usepackage[framemethod=tikz]{mdframed}
\usepackage{tikz-cd}

\newcommand{\N}{\mathbb{N}}

\newcommand{\spp}{\operatorname{\mathfrak{sp}}} 
   
\newcommand{\Sp}{\operatorname{\mathsf{Sp}}}              
\newcommand{\Aut}{\operatorname{\mathsf{Aut}}} 
 
\newcommand{\M}{\operatorname{\mathsf{M}}} 

\usepackage[framemethod=tikz]{mdframed}

\newcommand{\R}{\mathbb{R}}

\newcommand{\Spin}{\operatorname{\mathsf{Spin}}}     
\newcommand{\C}{\mathbb{C}}
\newcommand{\vl}{\mathbb{V}_{\lambda}}
\newcommand{\V}{\mathbb{V}}

\usepackage{xcolor}
\usepackage{amsmath,amssymb,amsthm}
\usepackage{nomencl}
\makenomenclature
\usepackage{enumerate}
\usepackage{textcomp,enumitem}
\newlist{mylist}{enumerate}{2}

\usepackage{epigraph}
\usepackage[toc]{glossaries}
\usepackage[font=small]{caption}
\usepackage{mathtools}

\usepackage{theoremref, enumitem}
\usepackage{hyperref}
\usepackage[all]{xy}
\usepackage{tikz-cd}
\usepackage{scrextend,mdwlist}
\renewcommand\labelenumi{\normalfont(\roman{enumi})}

\renewcommand\theenumi\labelenumi
\usepackage{pdfpages,enumerate}
\usepackage{epigraph}
\usepackage{fancybox}
\renewcommand{\phi}{\varphi}
\usepackage{url}

\usetikzlibrary{matrix}

\usetikzlibrary{calc}

\usepackage{xcolor}
\usetikzlibrary{matrix}
\usepackage{amsmath, amsthm, amssymb, amsfonts, enumerate, dlfltxbcodetips}
\newtheorem*{theorem*}{Theorem}
\newtheorem*{stelling*}{Stelling}
\newtheoremstyle{mystyle}
{6pt}
{6pt}
{\normalfont}
{0pt}
{\scshape\bfseries}
{.}
{4pt}
{}
\theoremstyle{mystyle}
\newtheorem{definition}{Definition}[section]
\newtheorem*{definition*}{Definition}

\newtheorem{remark}[definition]{Remark}

\newtheorem{example}[definition]{Example}

\renewcommand{\subset}{\subseteq}
\newtheoremstyle{mystyle3}
{6pt}
{6pt}
{\sffamily}
{0pt}
{\scshape\bfseries}
{.}
{4pt}
{}
\theoremstyle{mystyle3}

\newtheoremstyle{mystyle2}
{6pt}
{6pt}
{\itshape}
{0pt}
{\scshape\bfseries}
{.}
{4pt}
{}
\theoremstyle{mystyle2}
\newtheorem{theorem}[definition]{Theorem}

\newtheorem{prop}[definition]{Proposition}
\newtheorem{lemma}[definition]{Lemma}

\newtheorem{corollary}[definition]{Corollary}

\usepackage{hyperref}

\definecolor{headtitle}{RGB}{0,0,0} 
\colorlet{captionhead}{headtitle}

\def\Hrulefill{\leavevmode\leaders\hrule width 1ex height 0.7ex depth -0.6ex\hfill\kern0pt}

\usepackage{mdwlist}

\usepackage[english]{babel}

\begin{document}
\title[Models for irreducible representations of the symplectic algebra]{Models for irreducible representations of the symplectic algebra using Dirac-type operators}


\author{Guner Muarem}
\address{Campus Middelheim 
	\\Middelheimlaan 1
	\\M.G.221 
	\\2020 Antwerpen 
	\\	Belgi\"e}
\curraddr{}
\email{guner.muarem@uantwerpen.be\\ gmuarem@proximus.be}
\thanks{}



\date{27 September 2023}
\dedicatory{}
\maketitle

\begin{abstract}
\noindent In this paper we will study both the finite and infinite-dimensional representations of the symplectic Lie algebra $\mathfrak{sp}(2n)$ and develop a polynomial model for these representations. This means that we will associate  a certain space of homogeneous polynomials in a matrix variable, intersected with the kernel of $\mathfrak{sp}(2n)$-invariant differential operators related to the symplectic Dirac operator with every irreducible representation of $\mathfrak{sp}(2n)$.
We will show that the systems of symplectic Dirac operators can be seen as generators of parafermion algebras.
As an application of these new models, we construct a symplectic analogue of the Rarita-Schwinger operator using the theory of transvector algebras.
\end{abstract}
\maketitle
\section{Introduction}
\noindent 
In the paper by Van Lancker, Sommen \& Constales\ \cite{MR2106725} the authors constructed polynomial models for the orthogonal algebra $\mathfrak{so}(n)$. More specifically, they associated with every $\mathfrak{so}(n)$-irreducible representation $\vl$, a highest weight (vector) which was realised in terms of homogeneous polynomials in a matrix variable. 
Of course, this idea can be generalised for any Lie algebra $\mathfrak{g}$ and gives rise to the following situation:
Given an irreducible $\mathfrak{g}$-representation of highest weight $( \lambda _{1} ,\dotsc ,\lambda _{N})_{\mathfrak{g}}$ we want to associate a polynomial model in terms of multi-homogeneous polynomials of degree $(\lambda_1,\dots,\lambda_N)$ intersected with the kernel of some $\mathfrak{g}$-invariant differential operators $G_j$ for $j\in J$, where $J$ is a finite index set. Schematically, we are looking for the following association:
\begin{align*}
    \begin{matrix}
        \text{Irreducible $\mathfrak{g}$-representation}\\ (\lambda_1,\dots,\lambda_N)_{\mathfrak{g}}
    \end{matrix} \quad\longleftrightarrow\quad \begin{matrix}
        \text{Polynomial model}\\ 
        \mathcal{P}_{\lambda_1,\dots,\lambda_N}(\R^{2n\times N}\cap \ker(G_j)
    \end{matrix} 
\end{align*}
Note that the case of $\mathfrak{g}=\mathfrak{so}(n)$ was completely covered by the aforementioned paper. The setting for these polynomial models is the one of Clifford analysis: a hypercomplex function theory which revolves around the Dirac operator on $\R^n$. This operator acts on functions with values in the finite-dimensional spinor space $\mathbb{S}$ and is given by the expression $\up_x=\sum_{j=1}^n e_j\partial_{x_j}$, where $e_j$ are the so-called Clifford algebra generators satisfying $\{e_j,e_k\}=-2\delta_{jk}$.  Note that this means that the elements of the Clifford algebra are anti-commuting complex units. We refer the reader to \cite{delanghe2012clifford} for a thorough study of Clifford analysis.
The key property of the Dirac operator is that it factorises the Laplacian $\Delta=\sum_{j=1}^n\partial_{x_j}^2$ on $\R^{n}$ in the sense that $\up_x^2=-\Delta$. This leads to the observation that Clifford analysis is also a \textit{refinement} of harmonic analysis. It is a well-established fact in representation theory that the spaces $\mathcal{H}_k(\R^n,\C)$ of harmonic polynomials, which are homogeneous of degree $k$, form irreducible representations of the Lie algebra $\mathfrak{so}(n)$ of highest weight $(k,0,\dots,0)_{\mathfrak{so}(n)}$. The spinor representations (for $n$ even) of highest weight $(k+\frac{1}{2},\frac{1}{2},\dots,\pm \frac{1}{2})_{\mathfrak{so}(n)}$ can be modelled by the spaces $\mathcal{M}_k(\R^n,\mathbb{S})$ of $k$-homogeneous solutions of the Dirac operator known as \textit{monogenics}.
When working with (spinor-valued) functions depending on matrix variables one needs to work with systems of Dirac and Laplace operators (and associated operators, see later) which are related to the theory of Howe dualities \cite{howe1989remarks}. In \cite{MR2106725} (and already started in \cite{gilbert1991clifford}) this idea was used to associate with weights of the form $(\lambda_1,\dots,\lambda_N)_{\mathfrak{so}(n)}$ with dominant weight condition $\lambda_1\geq \cdots \lambda_N$ certain spaces of simplicial harmonics and monogenics (we come back to this in Section 2), related to these aforementioned systems of differential operators (see also \cite{colombo2012analysis}). The goal of this paper is to focus on the case where $\mathfrak{g}=\mathfrak{sp}(2n)$ in the setting of symplectic Clifford analysis, see \cite{de2017basic}.
The notion of the Dirac operator naturally generalises to the symplectic framework: the symplectic Dirac operator $D_s$ is now obtained as a contraction between the Weyl algebra generators $\{z_1,\dots,z_n,\partial_{z_1},\dots,\partial_{z_n}\}$ and the vector fields $\{\partial_{x_1},\dots,\partial_{y_n}\}$ using the canonical symplectic form $\omega_0$ on $\R^{2n}$. In other words, we can write $D_s=\langle \uz,\up_y\rangle -\langle \up_x,\up_z\rangle$, where $\langle \cdot,\cdot\rangle$ is the usual Euclidean inner product on $\R^n$. Note that this definition of the symplectic Dirac operator is slightly different than the one of \cite{de2017basic} but was already used in \cite{eelbode2022orthogonal,eelbode2023branching}.
We first of all stress that obtaining polynomial models for $\mathfrak{g}=\mathfrak{sp}(2n)$ will be slightly more intricate as the symplectic spinor space $\mS^{\infty}$ is infinite-dimensional. 
Therefore, we will work in two steps:
\begin{enumerate}
	\item We will first consider the finite-dimensional representations. The $\mathfrak{sp}(2n)$-invariant differential operators $G_j$ (from the scheme above) then will lead to the notion of \textit{symplectic simplicial harmonics} (see Section \ref{sec:simplicial_symp_harmonics}). In some sense, these results are anticipated by the general theory of Howe dualities (see \cite{howe1989remarks}).
	\item Afterwards, we will focus on an interesting class of infinite-dimensional representations of highest weight $(\lambda_1,\dots,\lambda_N)_{\mathfrak{sp}(2n)}\boxtimes\mathbb{S}^{\infty},$ where $\boxtimes$ denotes the Cartan product and $$\mS^{\infty}=\left(-\frac{1}{2},\dots,-\frac{1}{2}\right)_{\mathfrak{sp}(2n)}\oplus \left(-\frac{1}{2},\dots,-\frac{3}{2}\right)_{\mathfrak{sp}(2n)}.$$ These types of representations are connected with the theory of the symplectic Dirac operator. As a consequence, this will lead to the novel notion of \textit{symplectic simplicial monogenics} (see Section \ref{sec:simplicial_symp_monogenics}) and a connection with parastatistics.
\end{enumerate}
As an application, we will define `higher spin' Dirac operators (the Rarita-Schwinger operator describing particles with spin-$\frac{3}{2}$ is an example of such an operator) in the symplectic framework. We will call these \textit{higher metaplectic Dirac operators}. By means of example, we will introduce the symplectic Rarita-Schwinger operator in Section \ref{symplecticRS}.
\section{Simplicial harmonics and monogenics}
\subsection{The harmonic Fischer decomposition}
Consider the regular action of the Lie group $\mathsf{SO}(n)$ on the space of polynomials $\mathcal{P}(\R^n,\C)$:
\[ H : \mathsf{SO}(n) \to  \Aut(\mathcal{P}(\R^n,\C)), \quad g\mapsto\left( H(g)[P](\ux):= P(g^{-1}\ux)\right). \]
On the level of the Lie algebra $\mathfrak{so}(n)$, the so-called \textit{derived action} $dH$ on the polynomial space $\mathcal{P}(\R^n,\C)$ is given by the \textit{angular momenta operators} (known from quantum mechanics) which are defined by means of $L_{ab}=x_a\partial_{x_b}-x_b\partial_{x_a}$ for $1\leq a<b\leq n$.
In order to decompose the space of polynomials $\mathcal{P}(\R^n,\C)$ into irreducible representations for the orthogonal Lie algebra, we first note that the space can be written as a direct sum of $k$-homogeneous polynomials, i.e. $\mathcal{P}(\R^n,\C)=\bigoplus_{k=0}^{\infty}\mathcal{P}_k(\R^n,\C).$
The space of $k$-homogeneous polynomials $\mathcal{P}_k(\R^n,\C)$ further decomposes under the action of $\mathsf{SO}(n)$ in terms of solutions of the Laplacian on $\R^n$. 
The space of \textit{$k$-homogeneous harmonics} is denoted by $$\mathcal{H}_k(\R^n,\C):=\mathcal{P}_k(\R^n,\C)\cap \ker(\Delta).$$ 
\noindent Each of the spaces $\mathcal{H}_k(\R^n,\C)$ (for $n\geq 3)$ forms an irreducible $\mathfrak{so}(n)$-representation of highest weight $(k)_{\mathfrak{so}(n)}$ (see \cite{gilbert1991clifford}).
Using these $\mathfrak{so}(n)$-irreducible modules $\mathcal{H}_k(\R^n,\C)$ we can further decompose the space of $k$-homogeneous polynomials (see \cite{fischer1918differentiationsprozesse})
	\begin{align}\label{Fischerd}
 \mathcal{P}_k(\R^n,\C)=\bigoplus_{j=0}^{\lfloor\frac{k}{2}\rfloor}|\ux|^{2j}\mathcal{H}_{k-2j}(\R^n,\C). 
 \end{align}
The decomposition from above is not \textit{multiplicity-free}: every irreducible representation $\mathcal{H}_j(\R^n,\C)$ appears infinitely many times on each row. In order to make this decomposition multiplicity-free, one can use a \textit{hidden} symmetry (apart from the obvious $\mathsf{SO}(n)$-invariance). The multiplication operator $|\ux|^2$ allows us the raise the degree by two. On the other hand, the Laplace operator $\Delta$, reduces the degree by two. Lastly, the Euler operator $\mathbb{E}=\sum_{j=1}^n x_j\partial_{x_j}$ measures the degree of homogeneity of the corresponding $\mathcal{H}_k(\R^n,\C)$. This can be visualised as follows:
\[\begin{tikzcd}[every arrow/.append style={shift left}]
0&\arrow{l}{\Delta}\mathcal{H}_k\arrow[r,"\left|\ux\right|^2"] 
&\arrow{l}{\Delta}  \left|\ux\right|^2\mathcal{H}_k \ar[loop,"\mathbb{E}",swap,looseness=2]   \arrow[r,"\left|\ux\right|^2"] 
& \arrow{l}{\Delta} \left|\ux\right|^4\mathcal{H}_k \ar[loop,"\mathbb{E}",swap,looseness=2] \arrow[r,"\left|\ux\right|^2"] 
&  \ar[loop,"\mathbb{E}",swap,looseness=2] \arrow{l}{\Delta} \left|\ux\right|^6\mathcal{H}_k
\ \ \cdots
\end{tikzcd} \]
One can check that when these three operators are re-scaled as \begin{align*}
X := -\frac{1}{2}\Delta,
\qquad	Y: =\frac{1}{2}|\ux|^2,
\qquad	H :=[X,Y]= -\left(\mathbb{E}+\frac{n}{2}\right),
\end{align*}
we have $\Alg(X,Y,H)\cong\mathfrak{sl}(2)$. This algebra is the `hidden symmetry' which was mentioned above. As a result,  the space $\mathcal{P}(\R^n,\C)$ has a multiplicity-free decomposition under the \textit{joint action} of the (harmonic) Howe dual pair $\mathsf{SO}(n)\times \mathfrak{sl}(2)$.  
\subsection{The refined monogenic Fischer decomposition}\label{refinedmonogenicFD}
	Let $\{e_1,\dots,e_n\}$ be an orthonormal basis for $\R^n$ and let $\mathcal{F}(\R^n,\C)$ be a function space (e.g.\ the space of complex-valued polynomials). Then the \textit{Dirac operator} on $\R^n$ is defined as $\up_x=\sum_{j=1}^n e_j\partial_{x_j}$. 
The \textit{multiplication operator} is given by $
	\ux =\sum_{j=1}^n e_jx_j.
$
A solution $f$ of the Dirac operator $\up_x$ is called a \textit{monogenic} function. The space of all monogenic polynomials of degree $k\in\N$ is denoted \begin{align*}
	\mathcal{M}_k(\R^n,\mS) := \mathcal{P}_k(\R^n,\mS) \cap \ker \up_x.
	\end{align*}

We already mentioned that Clifford analysis is a \textit{refinement} of harmonic analysis. This refinement is also reflected in terms of the harmonic Howe duality $\mathsf{SO}(n)\times \mathfrak{sl}(2)$ found in the previous section. The Dirac operator is in a spin-invariant operator, so that the group part of the refined Howe duality is $G=\Spin(n)$. \par The `hidden symmetry' is then determined by the algebra generated by the Dirac operator $\up_x$ and its adjoint $\ux$. It is a well-known fact that is no longer a Lie algebra, but a Lie superalgebra $\mathfrak{osp}(1|2)$. This means that we obtain the \textit{refined} Howe duality $\Spin(n)\times \mathfrak{osp}(1|2)$ which governs the decomposition of the space $\mathcal{P}(\R^n,\C)\otimes \mathbb{S}$ of spinor-valued polynomials. However, we already know that the space of polynomials decomposes into the space of harmonics. This then leads to the following decomposition of the space of spinor-valued harmonics $\mathcal{H}_k(\R^n,\mathbb{S}):=\mathcal{H}_k(\R^n,\C)\otimes \mathbb{S}$:
\begin{theorem}[Monogenic refinement \cite{BDS}] \label{monogenicref}
	Let $H_k(\ux)\in\mathcal{H}_k(\R^n,\mathbb{S})$ be a $k$-homogeneous harmonic polynomial with spinor values. Then one has \[ H_k(\ux) = {M}_{k}(\ux)  + \ux {M}_{k-1}(\ux), \]
	where $M_j(\ux)\in\mathcal{M}_j(\R^n,\mathbb{S})$ (we ignore the parity of the spinors here). 
\end{theorem}
\noindent The definition of $k$-homogeneous monogenics can now be given an algebraic interpretation. The refinement theorem above can be formulated, purely in terms of $\mathfrak{so}(n)$ weights, as the tensor product (for $n$ even):
\[ \left(k\right)_{\mathfrak{so}(n)} \otimes \left(\frac{1}{2},\dots,\frac{1}{2},\pm\frac{1}{2}\right)_{\mathfrak{so}(n)}. \]
The tensor product rules for $\mathfrak{so}(n)$ (see for example \cite{frappat2000dictionary}) allow us to decompose this tensor product (up to isomorphism) as follows:
\begin{align*}
\left(k\right)_{\mathfrak{so}(n)} \otimes  \left(\frac{1}{2},\dots,\pm \frac{1}{2}\right)_{\mathfrak{so}(n)}\cong \left(k+\frac{1}{2},\dots,\pm \frac{1}{2}\right)_{\mathfrak{so}(n)}\oplus \left(k-\frac{1}{2},\dots,\mp \frac{1}{2}\right)_{\mathfrak{so}(n)}.
\end{align*}
In order to make the isomorphims  from above into an equality, we need the so-called \textit{embedding factors}. These are $\mathfrak{so}(n)$-invariant differential operators. The first component in the tensor product decomposition is trivially embedded using the identity operator, whereas the second component is embedded by the operator $\ux$. This now leads to the monogenic refinement of the Fischer decomposition as described in Theorem \ref{monogenicref}. In other words, $\mathcal{M}_k(\R^n,\mathbb{S})$ is exactly the Cartan product (for $n$ even) \[ \left(k\right)_{\mathfrak{so}(n)} \boxtimes \left(\frac{1}{2},\dots,\frac{1}{2},\pm\frac{1}{2}\right)_{\mathfrak{so}(n)}=\left(k+\frac{1}{2},\dots,\frac{1}{2},\pm\frac{1}{2}\right)_{\mathfrak{so}(n)}=(k)'_{\mathfrak{so}(n)}. \]
\par  We now obtain the full Fischer decomposition of the space of spinor-valued polynomials:
\begin{theorem}[Complete monogenic FD]\label{CompleteFD}
	Under the joint action of $\mathsf{Spin}(n)\times \mathfrak{osp}(1|2)$, the space of spinor valued polynomials decomposes as \[ \mathcal{P}(\R^n,\mathbb{S}) = \bigoplus_{k=0}^{\infty}\bigoplus_{j=0}^k\ux^j\mathcal{M}_{k-j}(\R^n,\mathbb{S}). \]
\end{theorem}

\subsection{Simplicial harmonic and monogenics}

In the previous section we considered polynomials in the vector variable $\ux=(x_1,\dots,x_n)\in\R^n$ and considered the multiplicity-free decomposition governed by the  Howe duality $\mathsf{SO}(n)\times \mathfrak{sl}(2)$. 
We will now work with $N$ vector variables 
$(\uu_1,\dots,\uu_N)\in\R^{n\times N}$. Note that this is also often called a \textit{matrix variable} (see Chapter 3, Section 4 in \cite{gilbert1991clifford} for more details), we denote it by $\uU=(\uu_1,\dots,\uu_N)\in\R^{n\times N}$.
The Howe duality $\mathsf{SO}(n)\times \mathfrak{sl}(2)$ which governs the multiplicity-free decomposition of the space of polynomials $\mathcal{P}(\R^n,\C)$ is only the first example of the general Howe duality $\mathsf{SO}(n)\times \mathfrak{sp}(2N)$ (see e.g.\ \cite{Goodman2009}) associated with the regular action $$H[g][P](\uU):=P(g^{-1}\uu_1,\dots,g^{-1}\uu_N)$$ of $\mathsf{SO}(n)$ on the space of polynomials $\mathcal{P}(\R^{n\times N},\C)$ in a matrix variable.

Note that in the spinor-valued case, the action is naturally extended by the multiplicative action of spinor space. 
The space $\mathcal{P}_{\lambda_1,\dots,\lambda_N}(\R^{n\times N},\C)$ of polynomials which are $\lambda_j$-homogeneous in each variable $\uu_j$ also has a decomposition in terms of $\mathsf{SO}(n)$-irreducible representations under the regular action above. The irreducible modules appearing here, are in fact a generalisation of the spaces $\mathcal{H}_k(\R^n,\C)$ from the previous section and are defined as follows (see \cite{MR2106725}):
\begin{definition}\label{simplicial} Let $N$ be an integer such that $1\leq N\leq \lfloor \frac{n}{2}\rfloor$.
	A function $f:\R^{n\times N}\to \C, (\uu_1,\dots,\uu_N)\mapsto f(\uu_1,\dots,\uu_N)$ is called \textit{simplicial harmonic} if it satisfies the system of equations \begin{align*}
	\langle \up_{u_i},\up_{u_j}\rangle f=0 &\quad \text{for all } 1\leq i,j\leq N\\
	\langle \uu_i,\up_{u_j}\rangle f=0 &\quad \text{for all } 1\leq i<j\leq N.
	\end{align*}
	The vector space of simplicial harmonic polynomials which are homogeneous of degree $\lambda_i$ in the vector variable $\underline{u}_i$ is denoted by $\mathcal{H}_{\lambda_1,\dots,\lambda_N}(\R^{n\times N},\C)$. Note that we must have $\lambda_1\geq \cdots \lambda_N\geq 0$ in order to satisfy the dominant weight condition.
\end{definition}

\begin{definition}\label{simplicialII}
	Let $N$ be an integer such that $1\leq N\leq \lfloor \frac{n}{2}\rfloor$.
	A spinor-valued function $f:\R^{n\times N}\to \mathbb{S}, (\uu_1,\dots,\uu_N)\mapsto f(\uu_1,\dots,\uu_N)$ is called \textit{simplicial monogenic} if it satisfies the system of equations \begin{align*}
	\up_{u_i}f =0&\quad \text{for all } 1\leq i\leq N\\
	\langle \uu_i,\up_{u_j}\rangle f=0 &\quad \text{for all } 1\leq i<j\leq N.
	\end{align*}
	The vector space of simplicial harmonic polynomials which are homogeneous of degree $\lambda_i$ in the vector variable $\underline{u}_i$ is denoted by $\mathcal{S}_{\lambda_1,\dots,\lambda_N}(\R^{n\times N},\mathbb{S})$. 
\end{definition}

\begin{theorem}[Van Lancker, Sommen \& Constales \cite{MR2106725}]\label{vlsctheorem}
	Let $N$ be an integer such that $1\leq N\leq \lfloor \frac{n}{2}\rfloor$ and suppose that the dominant weight condition is satisfied (for the Lie algebra $\mathfrak{so}(n))$. Then we have the following properties:
	\begin{enumerate}
		\item The spaces $\mathcal{H}_{\lambda_1,\dots,\lambda_N}(\R^n,\C)$ are irreducible of highest weight $(\lambda_1,\dots,\lambda_N)_{\mathfrak{so}(n)}$.
		\item The spaces $\mathcal{S}_{\lambda_1,\dots,\lambda_N}(\R^n,\C)$ are irreducible of highest weight $(\lambda_1,\dots,\lambda_N)'_{\mathfrak{so}(n)}$.
	\end{enumerate}
\end{theorem}
In \cite{brackx2011higher,de2010special} these models were proven to be useful to construct so-called higher spin Dirac operators. In the rest of this paper we will focus on the symplectic Lie algebra $\mathfrak{sp}(2n)$ instead of $\mathfrak{so}(n)$. This implies that the symmetric bilinear inner product on $\R^n$ will be replaced by a skew-symmetric one on $\R^{2n}$. More specifically, it has the following $(2n\times 2n)$-block matrix representation $$\Omega_0=\begin{pmatrix}
    0&I_n\\-I_n&0
\end{pmatrix}.$$
Later on, we will also use the notation \begin{align}\label{lrs}
    \langle \uv,\underline{w}\rangle_s:=\uv^T\Omega_0\underline{w}
\end{align}for all $\uv,\underline{w}\in \R^{2n}$.
As a consequence, the finite-dimensional Clifford algebra becomes an infinite-dimensional Weyl algebra and the Dirac operator is replaced by the symplectic Dirac operator. Due to the infinite-dimensionality of the symplectic spinors, this transition is rather intricate.\begin{remark}
    We further mention that there exist generalisations of the Fischer decomposition from equation (\ref{Fischerd}) and the refinement from Theorem \ref{monogenicref} for matrix variables, see \cite{lavivcka2019separation,van2012monogenic}.
\end{remark}

\section{The notion of symplectic simplicial harmonics}\label{sec:simplicial_symp_harmonics}
\noindent So far we discussed finite-dimensional irreducible representations of the Lie algebra $\mathfrak{so}(n)$ and characterised these modules in terms of simplicial harmonics and simplicial monogenics (see Theorem \ref{vlsctheorem}). We will now do exactly the same thing for the symplectic Lie algebra $\mathfrak{sp}(2n)$. We first focus on the finite-dimensional representations of $\mathfrak{sp}(2n)$. Recall that the finite-dimensional irreducible representations of the symplectic algebra are uniquely characterised by their labels satisfying $
\lambda_1\geq \lambda_2\geq \cdots\geq \lambda_n\geq 0$. 
In order to get some intuition with these finite-dimensional $\mathfrak{sp}(2n)$-irreducible representations, we start by calculating their dimension using the Weyl dimension formula (see \cite{MR3136522} for more background).
\begin{prop}
	Let $g_i:=n-i+1$ and $m_i:=\lambda_i+g_i$ for $1\leq i\leq n$ then \[ \dim(\lambda_1,\dots,\lambda_n)_{\mathfrak{sp}(2n)} =\prod_i \left(\frac{m_i}{g_i}\right) \prod_{i<j}\left(\frac{m_i-m_j}{g_i-g_j}\right)\prod_{i<j}\left(\frac{m_i+m_j}{g_i+g_j}\right).\]
\end{prop}
Using this formula, one can check that the dimension of module $(k)_{\mathfrak{sp}(2n)}$ is equal to \[ \dim (k)_{\mathfrak{sp}(2n)}=\binom{k+2 n-1}{2 n-1} . \]
	One might recognize this number as the dimension of the space of $k$-homgeneous polynomials $\mathcal{P}_{k}(\R^{2n},\C)$. This is no coincidence:
\begin{prop}\label{pkirrep}
	The space of $k$-homogeneous polynomials $\mathcal{P}_{k}(\R^{2n},\C)$ is an irreducible $\mathfrak{sp}(2n)$-module of weight $(k)_{\mathfrak{sp}(2n)}$ and the corresponding highest weight vector is $w_{k}(\ux)=x_1^{k}$.\end{prop}
\begin{proof}
	The generators of Lie algebra $\mathfrak{sp}(2n)$ can be realised on the space $\mathcal{P}(\R^{2n},\C)$ as the following differential operators:
	\begin{align}\label{sp(2n)proof}
	\begin{cases}
	Y_{jj}=	x_j\partial_{y_j} 						\qquad	&j=1,\dots,n 		\qquad n
	\\ Z_{jj}= y_j\partial_{x_j}					\qquad	&j=1,\dots,n       	\qquad n
	\\ X_{jk}=x_j\partial_{x_k} - y_k\partial_{y_j} 	\qquad	&j,k=1,\dots,n 	\qquad n^2
	\\Y_{jk}= x_j\partial_{y_k}+x_k\partial_{y_j} 	\qquad	&j<k=1,\dots,n \qquad \frac{n(n-1)}{2}
	\\ Z_{jk}= y_j\partial_{x_k}+y_k\partial_{x_j}	\qquad	&j<k=1,\dots,n	\qquad \frac{n(n-1)}{2}
	\end{cases}
	\end{align}
	The positive roots are $\Phi^+ = \{X_{jk},Y_{jk},Y_{jj}\mid 1\leq j<k\leq n\}$ and the Cartan algebra is given by $X_{jj}=x_j\partial_{x_j}-y_j\partial_{y_j}$ for $1\leq j\leq n$.
	One easily checks that $w_{k}=x_1^{k}$ is a solution of all the positive roots. Moreover, the action of the Cartan algebra $\mathfrak{h}$ on $w_{k}$  gives $H_1x_1^{k}=kx_1^{k}$ and $H_2x_1^{k}=\cdots=H_nx_1^{k}=0$. 
\end{proof}
\begin{remark}
	At this point it is worthwhile to compare the orthogonal and symplectic framework. In the orthogonal setting, the space of $k$-homogeneous polynomials is \textit{reducible} as an $\mathfrak{so}(n)$-module and further decomposes into harmonic polynomials (the latter are irreducible modules). In other words, we have \begin{align*}
	(k)_{\mathfrak{so}(n)} \longleftrightarrow \mathcal{P}_{k}(\R^n,\C)\cap \ker(\Delta).
	\end{align*}
	In the symplectic framework, we have
	\begin{align*}
	(k)_{\mathfrak{sp}(2n)} \longleftrightarrow \mathcal{P}_{k}(\R^{2n},\C).
	\end{align*}
	In some sense, this is to be expected as the Laplace operator $\Delta$ can be written as $\langle \up_x,\up_x\rangle$ whereas the natural operator in the symplectic setting would be $\langle \up_x,\up_x\rangle_s$ which is of course zero, so that no extra condition is to be imposed.
\end{remark}
\begin{remark}
Note that the remark above implies that there is no natural way to associate a second-order differential operator (the Laplacian) to the symplectic Dirac operator $D_s$ on $\R^n$. However, in \cite{habermann2006introduction} it was shown that if one allows a compatible complex structure $\mathbb{J}$ on $\mathbb{R}^{2n}$ one can define a second symplectic Dirac operator $D_t$ (which is basically a rotation of the operators $D_s$ using the complex structure) such that $[D_s,D_t]=-i\Delta$.
Recently, this led to a hermitian refinement of symplectic Clifford analysis in which the notion of complex harmonics was used to describe the solution space of the symplectic Dirac operators $D_s$ and $D_t$ using a Fischer decomposition, see \cite{eelbode2023hermitian}.
\end{remark}
\subsection{The case of $N=2$}
We now focus on the case where $N=2$ which allows us to take a look at the more general module $(\lambda_1,\lambda_2)_{\mathfrak{sp}(2n)}$ with $\lambda_1\geq \lambda_2\geq 0$. We denote the two vector variables $(\ux,\uu)\in \R^{2n\times 2}$ by
\begin{align*}
\ux&=(x_1,\dots,x_n,y_1,\dots,y_n)\in\R^{2n}\\
\uu&   =(u_1,\dots,u_n,v_1,\dots,v_n)\in\R^{2n}.
\end{align*} 
At this point, one might expect that 	$(\lambda_1,\lambda_2)_{\mathfrak{sp}(2n)}\leftrightarrow \mathcal{P}_{\lambda_1,\lambda_2}(\R^{2n},\C)$. This is however not the case as $\mathcal{P}_{\lambda_1,\lambda_2}(\R^{2n},\C)$ is \textit{not} irreducible as $\mathfrak{sp}(2n)$-module, as illustrated in the following counterexample:
\begin{example}\label{counterexample}
	We start by calculating the dimension of the space $\mathcal{P}_{\lambda_1,\lambda_2}(\R^{2n},\C)$: \begin{align*}
	\dim(\mathcal{P}_{\lambda_1,\lambda_2}(\R^{2n},\C)) = {\lambda_1+2n-1\choose 2n-1}{\lambda_2+2n-1\choose 2n-1}.
	\end{align*} 
	Let us now take $n=4$ and consider by means of example the module $(2,1)_{\spp(8)}$. The dimension of the polynomial algebra is given by 
	\[ \dim(\mathcal{P}_{2,1}(\mathbb{R}^4,\C)) = {9\choose 7}{8\choose 7}=288. \]
	However, one can check that $\dim(2,1)_{\spp(8)}=160$. This means that the suggested model $\mathcal{P}_{\lambda_1,\lambda_2}(\R^{2n},\C)$ is too big and some kind of \textit{reduction} has to be done.
\end{example}The tool for this dimensional reduction is the Howe duality $\Sp(2n)\times \mathfrak{so}(4)$.
\begin{prop}\label{N=2}
	Let $N=2$ and $(\ux,\uu)\in\R^{2n\times 2}$. Then the irreducible representation given by the weight $(\lambda_1,\lambda_2)_{\mathfrak{sp}(2n)}$ can be realised in terms of homogeneous polynomials as follows:
	\[ \mathcal{P}_{\lambda_1,\lambda_2}(\R^{2n\times 2},\C)\cap \ker(\langle \ux,\underline{\partial}_u\rangle, \langle\underline{\partial}_x,\underline{\partial}_u\rangle_s)\]
	where the operators $\langle \ux,\underline{\partial}_u\rangle$ and $\langle\underline{\partial}_x,\underline{\partial}_u\rangle_s$ are $\spp(2n)$-invariant operators and $\langle\cdot,\cdot\rangle_s$ is the `symplectic inner product' (see equation \ref{lrs}).
\end{prop}
\begin{proof}
	We construct $\mathfrak{sp}(2n)$-invariant operators by contraction of the coordinates and corresponding derivatives by using $\langle \cdot,\cdot\rangle $ or  $\langle \cdot,\cdot\rangle_s$.
	By a paring argument, we expect ${4\choose 2}=6$ symplectic invariant operators. They are given by: $\langle \underline{u},\underline{x}\rangle_s$, $\langle\underline{u},\up_x\rangle$, $\langle \underline{x},\up_u\rangle$, $\langle \up_x,\up_u\rangle_s$, $\langle \underline{u},\up_u\rangle+n$ and  $\langle \underline{x},\up_x\rangle +n$.
	Note that for the mixed contractions we use the usual Euclidean inner product, because otherwise the resulting operator is not symplectic invariant. These six $\spp(2n)$-invariant operators span a copy of the orthogonal algebra $\mathfrak{so}(4)$ (see for example Wallach \cite{Goodman2009}, Theorem 5.6.14).
Recall that Cartan algebra $\mathfrak{h}\subset \mathfrak{so}(4)$ is given by \[
	\mathfrak{h}=\begin{cases}
	H_1 = \langle \underline{u},\up_u\rangle + n
	\\ H_2 = \langle \underline{x},\up_x\rangle + n
	\end{cases}\]
 We now claim that the highest weight vector is given by \[ {w_{\lambda_1,\lambda_2}(\underline{x},\underline{u})=x_1^{\lambda_1-\lambda_2}(x_1u_2-x_2u_1)^{\lambda_2}}.  \]
	Note that this highest weight vector is independent of $(\uv,\uy)$.
	To confirm this, we check the following:
	\begin{enumerate}
		\item The vector $w_{\lambda_1,\lambda_2}$ is annihilated by the positive roots of the $\mathfrak{so}(4)$. We have \begin{align*}
		x_1 {\partial_{u_1}}{ \left(x_1^{\lambda _1-\lambda _2} \left(u_2 x_1-u_1 x_2\right){}^{\lambda _2}\right)} &= -\lambda _2 x_2 x_1^{\lambda _1-\lambda _2+1} \left(u_2 x_1-u_1 x_2\right){}^{\lambda _2-1}
		\\ x_2 {\partial_{u_2}}{ \left(x_1^{\lambda _1-\lambda _2} \left(u_2 x_1-u_1 x_2\right){}^{\lambda _2}\right)} &= \lambda _2 x_2 x_1^{\lambda _1-\lambda _2+1} \left(u_2 x_1-u_1 x_2\right){}^{\lambda _2-1}
		\end{align*}
		together with the fact that $w_{\lambda_1,\lambda_2}$ is independent of the variables $u_3,\dots,u_n$ and $v_3,\dots,v_n$ we can conclude that $	\langle \underline{x},\partial_{\underline{u}}\rangle w_{\lambda_1,\lambda_2}(\underline{x},\underline{u})=0$. One now similarly checks that $\langle \up_x,\up_u\rangle_s w_{\lambda_1,\lambda_2}(\ux,\uu)=0$.
		\item Moreover, $w_{\lambda_1,\lambda_2}$ must be annihilated by the positive roots vectors $\Phi^+_{\spp(2n)}$. A representation of the Lie algebra $\mathfrak{sp}(2n)$ on the space $\mathcal{P}(\R^{2n\times 2},\C)$ is given by:
		\begin{align*}
		\begin{cases}
		X_{jk}=x_j\partial_{x_k} - y_k\partial_{y_j}+u_j\partial_{u_k}-v_k\partial_{v_j} 	\qquad	&j,k=1,\dots,n 	 
		\\Y_{jk}= x_j\partial_{y_k}+x_k\partial_{y_j}+u_j\partial_{v_k}+u_k\partial_{v_j}	\qquad	&j<k=1,\dots,n
		\\ Z_{jk}= y_j\partial_{x_k}+y_k\partial_{x_j}+u_j\partial_{v_j}	\qquad	&j<k=1,\dots,n	
		\\ 	Y_{jj}=	x_j\partial_{y_j}+u_j\partial_{v_j} 						\qquad	&j=1,\dots,n 	
		\\ Z_{jj}= y_j\partial_{x_j}+v_j\partial_{x_j}					\qquad	&j=1,\dots,n    
		\end{cases}
		\end{align*}
		Due to the fact that $w_{\lambda_1,\lambda_2}$ only depends on the variables $(x_1,x_2,u_1,u_2)$, we immediately see that $w_{\lambda_1,\lambda_2}$ is a solution of the differential operators $Y_{jj}$	for all $1\leq j\leq n$.
		For the operators $X_{jk}$ with $1\leq j<k\leq n$ we have
		\begin{align*}
		X_{12}w_{\lambda_1,\lambda_2}(\ux,\uu)&=(x_1\partial_{x_2}-y_2\partial_{y_1}+u_1\partial_{u_2}-v_2\partial_{v_1})w_{\lambda_1,\lambda_2}(\ux,\uu)\\&=x_1\partial_{x_2}+u_1\partial_{u_2}w_{\lambda_{1,2}}(\ux,\uu)\\&=-u_1x_1x_1^{\lambda_1-\lambda_2}+u_1x_1x_1^{\lambda_1-\lambda_2}\\&=0
		\end{align*}
		and completely similar for the other operators $X_{jk}$. Repeating the strategy from above, we also obtain $Y_{jk}w_{\lambda_1,\lambda_2}=0$ for all $1\leq j<k\leq n$.  
		\item The action of the Cartan algebra $\mathfrak{h}=\{X_{jj}:j=1,\dots,n\}$ of $\spp(2n)$ must give the eigenvalues $(\lambda_1,\lambda_2)_{\mathfrak{sp}(2n)}$ (recall that the redundant zeros are omitted here). For ${X}_{11}$ we obtain
		\begin{align*}
		x_1\left((\lambda_1-\lambda_2)x_1^{\lambda_1-\lambda_2-1}(x_1u_2-x_2u_1)^{\lambda_2}+x_1^{\lambda_1-\lambda_2}\lambda_2(x_1u_2-x_2u_1)^{\lambda_2-1}u_2\right)\\+u_1\left(x_1^{\lambda_1-\lambda_2}\lambda_2(x_1u_2-x_2u_1)^{\lambda_2-1}(-x_2)\right)
		\end{align*}
		This can be reduced to $\lambda_1w_{\lambda_1,\lambda_2}(\underline{x},\underline{u})$ as wanted. By computing the 
		action of $X_{22}$ on $w_{\lambda_1,\lambda_2}(\underline{x},\underline{u})$ we obtain 
		\[ \lambda _2 u_2 x_1^{\lambda _1-\lambda _2+1} \left(u_2 x_1-u_1 x_2\right){}^{\lambda _2-1}-\lambda _2 u_1 x_2 x_1^{\lambda _1-\lambda _2} \left(u_2 x_1-u_1 x_2\right){}^{\lambda _2-1} = \lambda_2 w_{\lambda_1,\lambda_2}. \] Moreover, $X_{jj}w_{\lambda_1,\lambda_2}=0$ for $j>2$. 
	\end{enumerate}
 \end{proof}
We now come to the notion of simplicial harmonics in the symplectic framework for $N=2$ (compare with Definition \ref{simplicial}).
\begin{definition}\label{N2}
	For $\lambda_1\geq \lambda_2\geq 0$, we define the space of \textit{symplectic simplicial harmonics} as	\[ \mathcal{H}^s_{\lambda_1,\lambda_2}(\R^{2n\times 2},\C):=\mathcal{P}_{\lambda_1,\lambda_2}(\R^{2n\times 2},\C)\cap \ker(\langle \ux,\underline{\partial}_u\rangle, \langle\underline{\partial}_x,\underline{\partial}_u\rangle_s).\]
\end{definition}
\subsection{The general case of $N$ vector variables}
In this subsection, we generalise Proposition \ref{N=2} to the case of $N$ vector variables. We start by determining the associated Howe dual partner:
\begin{lemma}
	Let $\uU=(\uu_1,\dots,\uu_N)\in \R^{2n\times N}$.
	There are ${2N\choose 2}$ symplectic invariant operators which are obtained by contraction between
	\begin{enumerate}
		\item variables: $\langle \uu_i,\uu_j\rangle_s$ with $i,j\in\{1,\dots,N\}$ and $i\neq j$;
		\item differential operators: $\langle \underline{\partial}_{u_i},\underline{\partial}_{u_j}\rangle_s$ with $i,j\in\{1,\dots,N\}$ and $i\neq j$;
		\item variables and differential operators: $\langle \uu_i,\underline{\partial}_{u_j}\rangle$ with $i,j\in\{1,\dots,N\}$.
	\end{enumerate}
	These $\spp(2n)$-invariant operators give rise to a copy of the Lie algebra $\mathfrak{so}(2N)$. 
\end{lemma}
\begin{proof}
	This follows from straightforward computations of commutators. See for instance \cite{Goodman2009}.
\end{proof}
\noindent We now have the following generalisation of Definition \ref{N2}:

\begin{definition}
	We define the space of \textit{symplectic simplicial harmonics} as \[ \mathcal{H}^s_{\lambda_1,\dots,\lambda_N}(\R^{2n\times N},\C):= \mathcal{P}_{\lambda_1,\dots,\lambda_N}(\R^{2n\times N},\C)\cap \ker(\langle \uu_r,\up_{u_s}\rangle, \langle \up_{u_p},\up_{u_q}\rangle_s ) \]
	for $1\leq r<s\leq N$ and $1\leq p\neq q\leq N$.
\end{definition}
\noindent We now check that the space of symplectic simplicial harmonics provides a polynomial model for the irreducible representations of highest weight $(\lambda_1,\dots,\lambda_N)_{\mathfrak{sp}(2n)}$.
\label{scalar_kvector}
\begin{theorem}[Scalar polynomial model]\label{scalarsymplectic}
	The spaces $\mathcal{H}^s_{\lambda_1,\dots,\lambda_N}(\R^{2n\times N},\C)$ provide a model for the finite-dimensional representations of $\mathfrak{sp}(2n)$ of highest weight $ (\lambda_1,\dots,\lambda_N)_{\mathfrak{sp}(2n)}$.
\end{theorem}
\begin{proof}
	The proof is completely analogous to the proof of Proposition \ref{N=2}. However, we first note that $\mathfrak{sp}(2n)$ has the following representation (under the regular action) on the space of polynomials $\mathcal{P}(\R^{2n\times N},\C)$:
	\begin{align*}
	\begin{cases}
	Y_{ii}=\sum_{N=1}^n x_{N,i}\partial_{y_{N,i}}							\qquad	&i=1,\dots,n 	 
	\\ Z_{ii}= \sum_{N=1}^n y_{N,i}\partial_{x_{N,i}}					\qquad	&i=1,\dots,n        
	\\ X_{ij}=\sum_{N=1}^nx_{N,i}\partial_{x_{N,j}} - y_{N,j}\partial_{y_{N,i}} 	 	\qquad	&i,j=1,\dots,n 	 
	\\Y_{ij}= \sum_{N=1}^n x_{N,i}\partial_{y_{N,N}}+x_{N,j}\partial_{y_{N,i}} 	\qquad	&i<j=1,\dots,n 
	\\ Z_{ij}= \sum_{N=1}^n y_{N,i}\partial_{x_{N,j}}+y_{N,j}\partial_{x_{N,i}}	\qquad	&i<j=1,\dots,n 
	\end{cases}
	\end{align*}
	\noindent We now define $$\Xi_{N}:=\begin{pmatrix}
	x_{11} & x_{12} & \cdots & x_{1N}\\
	x_{21} & x_{22} & \cdots & x_{2N}\\
	\vdots & \vdots & \ddots & \vdots\\
	x_{N1} & x_{N2} & \cdots & x_{NN}\\
	\end{pmatrix}\in\mathbb{R}^{N\times N}$$ 
	The highest weight vector for $(\lambda_1,\dots,\lambda_N)_{\mathfrak{sp}(2n)}$ is then given by \begin{align*}
	w_{\lambda_1,\dots,\lambda_N}:&=\Xi_1^{\lambda_1-\lambda_2}\det(\Xi_2)^{\lambda_2-\lambda_3}\cdots \det(\Xi_N)^{\lambda_N}
	=\prod_{j=1}^N \det(\Xi_j)^{\lambda_j-\lambda_{j+1}},
	\end{align*}
	where $\lambda_j=0$ for $j>N$. Proceeding like in Proposition \ref{N=2}) now concludes the proof.
\end{proof}

\section{Symplectic simplicial monogenics}\label{sec:simplicial_symp_monogenics}
Now that we know that the spaces of symplectic simplicial harmonics
$\mathcal{H}^s_{\lambda_1,\dots,\lambda_N}(\R^{2n\times N},\C)$ form a polynomial model for finite-dimensional irreducible $\mathfrak{sp}(2n)$ representations of highest weight $(\lambda_1,\dots,\lambda_N)_{\mathfrak{sp}(2n)}$ we consider a class of infinite-dimensional $\mathfrak{sp}(2n)$-representations.
This special class of representations is induced by the symplectic Dirac operator which we will now define. We start with the symplectic version of the Clifford algebra:
\begin{definition}
    Let $(V,\omega)$ be a symplectic vector space.
The \textit{symplectic Clifford algebra} $\mathsf{Cl}_s(V,\omega)$ is defined as the quotient algebra of the tensor algebra $T(V)$ of $V$, by the two-sided ideal $$\mathcal{I}(V,{\omega}):=\{v\otimes u-u\otimes v+\omega(v,u) : u,v\in V\}.$$ In other words
$\mathsf{Cl}_s(V,\omega):=T(V)/\mathcal{I}(V,{\omega})$
is the algebra generated by $V$ in terms of the relation $[v,u]=-\omega(v,u)$, where we have omitted the tensor product symbols.
\end{definition}
\begin{remark}
    Note that the symplectic Clifford is often introduced as the Weyl algebra generated by $2n$ commuting variables and their associated partial derivatives where e.g. $[\partial_{z_i},z_j] = \delta_{ij}$.

\end{remark}
\subsection{The symplectic Dirac operator} \label{fockdirac}
Following Habermann \cite{habermann2006introduction} on the flat symplectic space $(\R^{2n},\omega_0)$, we define the \textit{symplectic Dirac operator} by means of \begin{align*}
    D_s : \mathcal{P}(\R^{2n},\C)\otimes \mS^{\infty} \to \mathcal{P}(\R^{2n},\C)\otimes \mS^{\infty}, \quad f\mapsto (\langle \underline{z},\underline{\partial}_y\rangle - \langle \underline{\partial}_x,\underline{\partial}_z \rangle)f,
\end{align*}
where $\langle\cdot,\cdot\rangle$ is the usual Euclidean inner product on $\R^n$.
Note that $\uz \in \R^n$ plays a different role than $(\ux,\uy) \in \R^{2n}$ here (it is used to model the symplectic spinors as polynomials in $\uz$). 
The adjoint (with respect to the symplectic Fischer inner product, see \cite{de2017basic}) is given by $X_s= \langle \underline{x},\underline{z}\rangle + \langle \underline{y},\underline{\partial}_z\rangle$.

\begin{lemma}[\cite{eelbode2022orthogonal}]\label{sprealisationI}
	The symplectic Lie algebra $\mathfrak{sp}(2n)$ has the following realisation on the space of symplectic spinor-valued polynomials $\mathcal{P}(\mathbb{R}^{2n},\mathbb{C})\otimes\mathcal{P}(\mathbb{R}^n,\C)$:
	\begin{align}\label{realisaties1}
	\begin{cases}
	X_{jk}=x_j\partial_{x_k}-y_k\partial_{y_j} - (z_k\partial_{z_j}+\frac{1}{2}\delta_{jk})
	&\quad j,k=1,\dots,n\quad\qquad n^2
	\\Y_{jk}=x_j\partial_{y_k}+x_k\partial_{y_j} -\partial_{z_j}\partial_{z_k}
	&\quad j<k=1,\dots,n\quad \frac{n(n-1)}{2}
	\\ Z_{jk}=y_j\partial_{x_k}+y_k\partial_{x_j} +  z_jz_k
	&\quad j<k=1,\dots,n\quad \frac{n(n-1)}{2}
	\\Y_{jj}=x_j \partial_{y_j} -\frac{1}{2}\partial_{z_j}^2
	&\quad j=1,\dots,n\qquad\qquad n
	\\Z_{jj}=y_j\partial_{x_j}+\frac{1}{2} z_j^2
	&\quad j=1,\dots,n\qquad\qquad n
	\end{cases}
	\end{align} 
	The Cartan algebra $\mathfrak{h}$ is given by the operators $X_{jj}$ for $1\leq j\leq n$ and the positive root vectors are $X_{jk},Y_{jk}$ and $Y_{jj}$ for $1\leq j<k\leq n$. 
\end{lemma}
\begin{corollary}
    The symplectic Dirac operator $D_s$ is a $\mathfrak{sp}(2n)$-invariant operator with respect to the realisation from Lemma \ref{sprealisationI}.
\end{corollary}
We define the space of $k$-homogeneous symplectic monogenics as \begin{align*}
    \mathcal{M}_k^s:=\ker_k(D_s)=\ker(D_s)\cap \left(\mathcal{P}_k(\R^{2n},\C)\otimes \mS^{\infty}\right).
\end{align*}
We will see in the next subsection that these spaces corresponds to the (direct sum) of the highest weights: \[ \left(k-\frac{1}{2},\dots,-\frac{1}{2}\right)_{\mathfrak{sp}(2n)}\oplus \left(k-\frac{1}{2},\dots,-\frac{3}{2}\right)_{\mathfrak{sp}(2n)}.  \]
Our goal will be to associate with the infinite-dimensional representation
\[ \left(\lambda_1-\frac{1}{2},\dots,\lambda_N-\frac{1}{2}\right)_{\mathfrak{sp}(2n)}\oplus \left(\lambda_1-\frac{1}{2},\dots,\lambda_N-\frac{3}{2}\right)_{\mathfrak{sp}(2n)}   \]
a certain space of $(\lambda_1,\dots,\lambda_N)$-homogeneous polynomials intersected with some $\mathfrak{sp}(2n)$-invariant differential operators (associated with the symplectic Dirac operator).
In order to obtain a characterisation of this space in terms of a Cartan product of two $\mathfrak{sp}(2n)$-irreducible representations, we need some more background. This is provided in the following subsection.

\subsection{Completely pointed modules}\label{compointed}
We will now introduce the algebraic background for computing tensor products of finite-dimensional $\mathfrak{sp}(2n)$-representations with the infinite-dimensional $\mathfrak{sp}(2n)$-module $\mathbb{S}^{\infty}$. We first recall some basic terminology considering the root system of $\mathfrak{sp}(2n)$. 
We denote by $\epsilon_i=(0,\cdots,1,\cdots,0)$ the standard basis of $\mathbb{R}^n$. The roots of the symplectic algebra $\mathfrak{sp}(2n)$ fall into two categories: $2n$ so-called \textit{long roots} $\pm 2\epsilon_i$ and $2(n^2-n)$ \textit{short roots} $\pm \epsilon_i\pm \epsilon_j$, where $1\leq i\leq n$ and $i\neq j$.
This means that the full set of roots $\Phi$ equals \[ \Phi=\{\pm(\epsilon_i\pm \epsilon_j):1\leq i< j \leq n\}\cup \{\pm 2\epsilon_i:1\leq i\leq n\}. \]
The simple roots are chosen to be \[ \Delta=\{\alpha_i:= \epsilon_i - \epsilon_{i+1}:  1\leq i\leq n-i\}\cup \{\alpha_n:= 2\epsilon_n\}. \]
The fundamental weights are now given by $
\omega_j = \epsilon_1+\cdots+\epsilon_j.
$
The two bases of the dual of the Cartan algebra $\mathfrak{h}^*$ relate as follows
\begin{align*}
\omega_1 = \epsilon_1,\quad	\omega_2 = \epsilon_1+\epsilon_2,\quad \dots\quad \omega_n = \epsilon_1+\epsilon_2+\cdots+\epsilon_n.
\end{align*}
We denote by  $\vl$ to be the irreducible highest weight module with highest weight $\lambda$. Note that we will sometimes denote this by $\V(\lambda)$ for some specific value of $\lambda$. 
We will now introduce a specific type of module of which the modules $\mS^{\infty}$ will appear to be examples.
\begin{definition} Let $\mathfrak{g}$ be a semisimple Lie algebra with Cartan algebra $\mathfrak{h}$ and let $\mathbb{V}$ be $\mathfrak{g}$-module.
	\begin{enumerate}
		\item We say that $\mathbb{V}$ is a module with \textit{bounded multiplicities} if there exists a $k\in\mathbb{N}_0$ such that for every decomposition of the form $	\V=\bigoplus_{\lambda} \vl,	\label{lambda}$ we have that $\dim \vl\leq k$. The minimal $k$ for which this upper bound holds is called the \textit{degree of $\V$} and it is denoted by $\deg(\V)=k$.
		\item The module $\V$ is then called \textit{completely pointed} if $\deg(\V)=1$.
	\end{enumerate}
\end{definition}
\noindent The Lie algebras $\mathfrak{g}$ which have infinite-dimensional modules with bounded multiplicities are rather limited as they only occur if and only if $\mathfrak{g}=\mathfrak{sl}(n)$ or $\mathfrak{g}=\spp(2n)$ (see for example \cite{MR1297597}). The following theorem restricts the possibilities for the completely pointed modules in the case of symplectic algebra. 
\begin{theorem}[Britten, Hooper \& Lemire \cite{MR1297597}] Let $\vl$ be a highest weight module of the symplectic Lie algebra $\mathfrak{sp}(2n)$. Then it is completely pointed if and only if $\lambda=0,\omega_1, -\frac{1}{2}\omega_n$ or $\omega_{n-1}-\frac{3}{2}\omega_n$ (or $\omega_n$ if $n=2$ or $3$).
	\begin{enumerate}
		\item 	If $\lambda=0$ and $\lambda=\omega_1$ the module is finite-dimensional, more specifically: $\V_0=\mathbb{C}v$ is the trivial representation and $\V_{\omega_1}\cong \C^{2n}$ is the fundamental representation. Moreover, in the special case $n=2,3$ the module $\V_{\omega_n}$ is also finite-dimensional.
		\item $\V_{\lambda}$ for $\lambda=	-\frac{1}{2}\omega_n$ and $\lambda=	\omega_{n-1}-\frac{3}{2}\omega_n$ are infinite-dimensional $\mathfrak{sp}(2n)$-modules. Their highest weight are given by \[ \left(-\frac{1}{2},\dots,-\frac{1}{2}\right)_{\mathfrak{sp}(2n)}\quad\text{ and }\quad\left(-\frac{1}{2},\dots,-\frac{3}{2}\right)_{\mathfrak{sp}(2n)}. \]
	\end{enumerate}
\end{theorem}
\noindent These completely pointed modules have the following interpretation:
\begin{theorem} \label{sinftylabel}
	The infinite-dimensional symplectic spinor space $\mathbb{S}^{\infty}$ is a reducible $\mathfrak{sp}(2n)$-module, its irreducible components exactly coincide with the two completely pointed weight modules from above. The corresponding highest weight vector  is given by $1\oplus z_n$ (when considering the realisation from Lemma \ref{sprealisationI}).
\end{theorem}
\begin{proof}
	Recall from Lemma \ref{sprealisationI} that the $n$ Cartan elements (we only focus on the spinor part here) of the Lie algebra $\mathfrak{sp}(2n)$ are given by \[ H_j:=X_{jj}= - \left(z_j\partial_{z_j}+\frac{1}{2}\right) \qquad (1\leq j\leq n).\]
	We now directly verify that the completely pointed module  $\mathbb{V}\left(-\frac{1}{2}\omega_n\right)$ has highest weight vector $1$ and	$\mathbb{V}\left(\omega_{n-1}-\frac{3}{2}\omega_n\right)$ has highest weight vector $z_n$. This implies that the highest weight vector of $\mathbb{S}^{\infty}$ is $w=1\oplus z_n$.  
\end{proof}
%
\noindent The following theorem will allow us to compute the tensor product of finite-dimensional symplectic modules with the infinite-dimensional symplectic modules $\mathbb{S}^{\infty}_{\pm}$. 

\begin{theorem}[{Britten \& Lemire} \cite{bittenlemire}]\label{BrittenLemire}
	Let $\mathbb{V}(\lambda)$ be a finite-dimensional $\mathfrak{sp}(2n)$-irreducible representation of highest weight $(\lambda_1,\dots,\lambda_n)_{\mathfrak{sp}(2n)}$. Then, the tensor products with the infinite-dimensional completely pointed modules decompose as follows:
	\begin{align*} \mathbb{V}\left(-\frac{1}{2}\omega_n\right) \otimes \mathbb{V}(\lambda)  &\cong \bigoplus_{\kappa\in T_\lambda}\mathbb{V}\left(-\frac{1}{2}\omega_n+\kappa\right)\\
	\mathbb{V}\left(\omega_{n-1}-\frac{3}{2}\omega_n\right) \otimes \mathbb{V}(\lambda)  &\cong \bigoplus_{\kappa\in T_\lambda}\mathbb{V}\left(\omega_{n-1}-\frac{3}{2}\omega_n+\kappa\right)\end{align*}
	where $\kappa$ is of the following form
	\begin{align*}
	T_{\lambda} = \left\{\kappa\mid \kappa=\lambda - \sum_{i=1}^n d_i\epsilon_i\right\}
	\end{align*}
	subject to the properties:
	\begin{enumerate}
		\item $d_i\in\mathbb{N}$,
		\item $\sum_{i=1}^nd_i\in 2\mathbb{N}$,
		\item $0\leq d_i\leq \nu_i $ for $i=1,\dots,n-1$,
		\item $0\leq d_n\leq 2\nu_n+1$.
	\end{enumerate}
\end{theorem}
\subsection{The symplectic monogenic Fischer decomposition}\label{subsection:SFD}
\noindent As we know from Proposition \ref{pkirrep} that $(k)_{\mathfrak{sp}(2n)}\leftrightarrow\mathcal{P}_{k}(\R^n,\C)$ this now gives rise to the following result, which is in fact the symplectic counterpart of Theorem \ref{CompleteFD} proven 
in \cite{MR3126939}.

\begin{theorem}[Symplectic monogenic FD, \cite{MR3126939}]\label{SymFD}
	Under the joint action of $\mathsf{Mp}(2n)\times \mathfrak{sl}(2)$ the space of polynomials with values in the space of symplectic spinors decomposes as follows:
	\begin{align}\label{decomposition}
	\mathcal{P}(\mathbb{R}^{2n},\C)\otimes \mathbb{S}^{\infty} =\bigoplus_{k=0}^{\infty}\bigoplus_{j=0}^{\infty} X_s^j\mathcal{M}_{k}^{s},
	\end{align}
	where $\mathcal{M}_{k}^{s}= \ker(D_s) \cap (\mathcal{P}_{k}(\R^{2n},\C) \otimes \mathcal{P}(\R^n,\C))$.
\end{theorem}
 
\begin{remark}\label{ref_Mk}
	The module $\mathcal{M}^s_{k}$ is $\mathfrak{sp}(2n)$-reducible and decomposes into two irreducible parts (an even and odd part)
	$ \mathcal{M}_{\ell}^s = \left(\mathcal{M}_{\ell}^s\right)^+ \oplus \left(\mathcal{M}_{\ell}^s\right)^-$,
	which are irreducible $\mathfrak{sp}(2n)$-modules of highest weight
	\begin{align*}
	\left(\mathcal{M}_{k}^s\right)^+ &\leftrightarrow  \left(	k-\frac{1}{2},-\frac{1}{2},\dots,-\frac{1}{2}	\right)_{\mathfrak{sp}(2n)} \\ \left(\mathcal{M}_{k}^s\right)^- &\leftrightarrow \left(	k-\frac{1}{2},-\frac{1}{2},\dots,-\frac{3}{2}	\right)_{\mathfrak{sp}(2n)}.
	\end{align*}
\end{remark}
As in the orthogonal case, we now see that the space of $k$-homogeneous symplectic monogenics can be defined in a purely algebraic way:
\begin{definition}\label{Mks}
	The space of $k$-homogeneous polynomials solutions of the symplectic Dirac operator $D_s$ are called \textit{symplectic monogenics} and are denote by $\mathcal{M}_k^s$. They are algebraically determined by
	\[ \mathcal{M}_k^s := (k)_{\mathfrak{sp}(2n)}\boxtimes\left( \left(-\frac{1}{2},-\frac{1}{2},\dots,-\frac{1}{2}\right)_{\mathfrak{sp}(2n)}\oplus \left(-\frac{1}{2},-\frac{1}{2},\dots,-\frac{3}{2}\right)_{\mathfrak{sp}(2n)}\right)\]
\end{definition}
\begin{corollary}
	The $\mathfrak{sp}(2n)$-module $\mathcal{M}_k^s$ has highest weight vector $x_1^k\otimes (1\oplus z_n)$.
\end{corollary}
\begin{proof}
	This follows from Propositions \ref{pkirrep} and \ref{sinftylabel}.
\end{proof}
\subsection{The case of $N$ vector variables}\label{DsxDsu}
We will now be interested in systems of $N$ symplectic Dirac operators. We will begin by studying the case of $N=2$.
Recall from Section 3 that the \textit{finite-dimensional} representations of highest weight $(\lambda_1,\lambda_2)_{\mathfrak{sp}(2n)}$ are irreducible and can be explicitly realised as the space of symplectic simplicial harmonics. We now want to gain knowledge on the \textit{infinite-dimensional} irreducible representations of highest weight \[(\lambda_1,\lambda_2)'_{\mathfrak{sp}(2n)}:= \left(\lambda_1-\frac{1}{2},\lambda_2-\frac{1}{2},\dots,-\frac{1}{2}\right)_{\mathfrak{sp}(2n)}\oplus \left(\lambda_1-\frac{1}{2},\lambda_2-\frac{1}{2},\dots,-\frac{3}{2}\right)_{\mathfrak{sp}(2n)}. \]
As one might intuitively predict, the space $(\lambda_1,\lambda_2)'_{\mathfrak{sp}(2n)}$ will be defined by means of the Cartan product $(\lambda_1,\lambda_2)_{\mathfrak{sp}(2n)}\boxtimes \mathbb{S}^{\infty}$, but the question obviously arises as how to define these spaces as a kernel space for differential operators. We start by proving that the intuitive guess is indeed correct.
\begin{lemma}\label{cartanl1l2}
	The Cartan product of $(\lambda_1,\lambda_2)_{\mathfrak{sp}(2n)}\boxtimes \mathbb{S}^{\infty}$ is given by the direct sum of:
	\begin{align*}
	\V\left(-\frac{1}{2}\omega_n+\lambda_1\epsilon_1+\lambda_2\epsilon_2\right)&=\left(\lambda_1-\frac{1}{2},\lambda_2-\frac{1}{2},-\frac{1}{2},\dots,-\frac{1}{2}\right)_{\mathfrak{sp}(2n)}\\
	\V\left(\omega_{n-1}-\frac{3}{2}\omega_n+\lambda_1\epsilon_1+\lambda_2\epsilon_2\right)&=\left(\lambda_1-\frac{1}{2},\lambda_2-\frac{1}{2},-\frac{1}{2},\dots,-\frac{3}{2}\right)_{\mathfrak{sp}(2n)}
	\end{align*}
\end{lemma}
\begin{proof}
	In order to check this, we need to calculate (a part of) the tensor product decomposition of $(\lambda_1,\lambda_2)_{\mathfrak{sp}(2n)}\otimes \mathbb{S}^{\infty}$. This can be done using Theorem \ref{BrittenLemire} where $\mathbb{V}(\lambda)=(\lambda_1,\lambda_2)_{\mathfrak{sp}(2n)}$. 
	We use the Young convention, so in other words we write $ \lambda=\lambda_1\epsilon_1+\lambda\epsilon_2$.
	Then $\kappa$ from Theorem \ref{BrittenLemire} is given by \[ \kappa= (\lambda_1-d_1)\epsilon_1+(\lambda_2-d_2)\epsilon_2-d_n\epsilon_n \]
	Taking $d_1=d_2=d_{n}=0$ (as this will give the highest weight) leaves us with \[ \kappa=\lambda\epsilon_1+\lambda_2\epsilon_2. \]
	In order words 
	$\V\left(-\frac{1}{2}\omega_n+\kappa\right)\oplus$ $\V\left(\omega_{n-1}-\frac{3}{2}\omega_n+\kappa\right)$ is the Cartan product.
\end{proof}
\noindent 
We will now characterise the Cartan product  $(\lambda_1,\lambda_2)_{\mathfrak{sp}(2n)}\boxtimes \mathbb{S}^{\infty}$ as the polynomial kernel space of certain differential operators. This will be done by \textit{extending} the Howe dual partner $\mathfrak{so}(2N)$ from Section \ref{scalar_kvector}. We already know from the symplectic Fischer decomposition is that this dual parter is $\mathfrak{sl}(2)$ for $N=1$. As we are working with two vector variables here, we have two symplectic Dirac operators (one for each vector variable) and two adjoint operators which we define as follows:
\begin{align*}
D_{s,x} = \langle \uz,\up_y\rangle-\langle \up_z,\up_x\rangle \qquad\text{and}\qquad X_{s,x}= \langle\uy,\up_z\rangle + \langle\ux, \uz\rangle
\\		D_{s,u}  = \langle z,\up_v\rangle -\langle \up_z,\up_u\rangle  \qquad\text{and}\qquad 	X_{s,u}=\langle v,\up_z\rangle+ \langle u,\uz \rangle 
\end{align*}
Recall that $\langle \cdot,\cdot \rangle$ is the usual Euclidean inner product here. We now turn to the question whether these operators close as a Lie algebra.

\begin{lemma}\label{so5}
	The algebra generated by the Dirac operators $D_{s,x}$ and $D_{s,u}$ and their adjoints $X_{s,x}$ and $X_{s,u}$ gives rise to a copy of the Lie algebra $\mathfrak{so}(5)$.
\end{lemma}
\begin{proof}
	This follows from calculating the commutators between these operators, we have for instance:
	\begin{align*}
	[X_{s,u}, D_{s,u}]= \mathbb{E}_u+\mathbb{E}_v+n
	&& [X_{s,x}, D_{s,x}]= \mathbb{E}_x+\mathbb{E}_y+n
	&& [X_{s,u}, D_{s,x}]=  \langle \uu,\partial_{\ux}\rangle 
	\\ [X_{s,x}, D_{s,u}]= \langle \ux,\partial_{\uu}\rangle 
	&& [D_{s,u},D_{s,x}]=\langle \partial_{\ux},\partial_{\uu}\rangle_s
	&& [X_{s,u},X_{s,x}]=\langle {\ux},{\uu}\rangle_s
	\end{align*}
\end{proof}

This now leads to the following generalisation of symplectic monogenics:
\begin{definition}\label{ssm}
	We define the space of \textit{symplectic simplicial monogenics} for $N=2$ as follows:
	\begin{align*}
	{\mathcal{S}_{\lambda_1,\lambda_2}^s(\mathbb{\mathbb{R}}^{4n},\mathbb{S}^{\infty}) := \mathcal{P}_{\lambda_1,\lambda_2}(\mathbb{\mathbb{R}}^{4n},\mathbb{S}^{\infty}) \cap \ker(D_{s,x}, D_{s,u},\langle \ux,\up_u\rangle,\langle \up_x,\up_x\rangle_s ).  }
	\end{align*}
\end{definition}
\noindent We obtain the following characterisation of the Cartan product $(\lambda_1,\lambda_2)_{\mathfrak{sp}(2n)}\boxtimes \mathbb{S}^{\infty}$:
\begin{theorem}\label{spinor_kvector}
	The space of symplectic simplicial monogenics $	\mathcal{S}_{\lambda_1,\lambda_2}^s(\mathbb{\mathbb{R}}^{4n},\mathbb{S}^{\infty})$ is the direct sum of two irreducible $\mathfrak{sp}(2n)$-module of highest weight 	\[ \left(\lambda_1-\frac{1}{2},\lambda_2-\frac{1}{2},-\frac{1}{2},\dots,-\frac{1}{2}\right)_{\mathfrak{sp}(2n)} \oplus \left(\lambda_1-\frac{1}{2},\lambda_2-\frac{1}{2},-\frac{1}{2},\dots,-\frac{3}{2}\right)_{\mathfrak{sp}(2n)}\]
	with highest weight vector $x_1^{\lambda_1-\lambda_2}(x_1u_2-x_2u_1)^{\lambda_2} \otimes (1\oplus z_n)$.
\end{theorem}
\begin{proof}
	The highest weight vector is obtained as the tensor product of the highest weight vector from Theorem \ref{scalarsymplectic}(iii) (scalar part) and  $w_{\lambda_1,\lambda_2}(\underline{x},\underline{u})\otimes (1\oplus z_n)$. The rest is a straightforward computation.
	%
\end{proof}
\begin{remark}
	In order to obtain the full decomposition of $(\lambda_1,\lambda_2)_{\mathfrak{sp}(2n)}\otimes \mathbb{S}^{\infty}$ (this will lead to a new Fischer decomposition, which has to be interpreted as the $N=2$ generalisation of \ref{SymFD}) we need to take into account the other values for the parameters $d_1,d_2,d_n\in \N$ (see Lemma \ref{cartanl1l2}).
	{ One must pay some attention here: in the $N=1$ case, the space $\mathcal{P}_{k}(\R^{2n},\C)$ is $\mathfrak{sp}(2n)$-irreducible and the operator $X_s$ (and its powers) served as the symplectic embedding factors making the isomorphism into an equality. In the case of two vector variables, the space $\mathcal{P}_{\lambda_1,\lambda_2}(\R^{4n},\C)$ is \textit{not} irreducible (see Example \ref{counterexample}). By the general theory of Howe dualities, this space further decomposes into the symplectic simplicial harmonics $\mathcal{H}_{\lambda_1,\lambda_2}^s(\R^{4n},\C)$. Thus, when $N>1$ the `analogy' with the classical (orthogonal case) is restored in a sense:
		the space $\mathcal{P}_{\lambda_1,\dots,\lambda_N}(\R^{2nN},\C)$ has a 
		`symplectic harmonic' Fischer decomposition with \textit{symplectic harmonic Howe dual pair} $\mathfrak{sp}(2n)\times \mathfrak{so}(2N)$. One could then look for a `monogenic refinement' which is then governed by the Howe dual pair $\mathfrak{sp}(2n)\times \mathfrak{so}(2N+1)$. For example, in the case of $N=2$ the latter would result in computing the tensor product 
		$\mathcal{H}^s_{\lambda_1,\lambda_2}(\R^{4n},\C)\otimes \mathbb{S}^{\infty}$
		where we already know that 
		$\mathcal{H}^s_{\lambda_1,\lambda_2}(\R^{4n},\C)\boxtimes \mathbb{S}^{\infty}=\mathcal{S}^s_{\lambda_1,\lambda_2}(\R^{4n},\C)$. This would yield the `symplectic translations' of the results from \cite{van2012monogenic, lavivcka2019separation}. We plan to do this in a subsequent paper.
	}
\end{remark}
In the case of $N$ vector variables we obtain:
\begin{lemma}\label{cartanLN}
	The Cartan product of $(\lambda_1,\lambda_2,\dots,\lambda_N)_{\mathfrak{sp}(2n)}\boxtimes \mathbb{S}^{\infty}$ is given by the direct sum of:
	\begin{align*}
	\V\left(-\frac{1}{2}\omega_n+\sum_{j = 1}^N\lambda_j\epsilon_j\right)&=\left(\lambda_1-\frac{1}{2},\lambda_2-\frac{1}{2},\dots,\lambda_N-\frac{1}{2}\right)_{\mathfrak{sp}(2n)}\\
	\V\left(\omega_{n-1}-\frac{3}{2}\omega_n+\sum_{j = 1}^N\lambda_j\right)&=\left(\lambda_1-\frac{1}{2},\lambda_2-\frac{1}{2},\dots,\lambda_N-\frac{3}{2}\right)_{\mathfrak{sp}(2n)}
	\end{align*}
\end{lemma}
\begin{proof}
	This is again a consequence of Theorem \ref{BrittenLemire} where one now takes $\mathbb{V}(\lambda)=(\lambda_1,\lambda_2,\dots,\lambda_N)_{\mathfrak{sp}(2n)}$ and $\kappa=\sum_{j=1}^N\lambda_j\epsilon_j$.
\end{proof}
\begin{theorem}\label{so2k+1}
	Considering a system of $N$ symplectic Dirac operators and their (Fischer) adjoints. These operators give rise to a copy of the orthogonal Lie algebra $\mathfrak{so}(2N+1)$ of dimension $N (1 + 2 N)$.
\end{theorem}
\begin{proof}
	We use the short-hand notation $D_a=D_{s,u_a}$ and $X_a=X_{s,u_a}$ for $1\leq a\leq N$. Then, the following commutation relations hold:
	\begin{align*}
	[[D_a,X_b],X_{c}] &= -2\delta_{ac}X_b\\
	[[D_a,X_b],D_{c}]&=2\delta_{bc}X_a\\
	[[D_a,D_b],X_{c}] &= 2\delta_{bc}D_a - 2\delta_{ac}D_b\\
	[[X_a,X_b],X_{c}] &= 2\delta_{bc}X_a - 2\delta_{ac}X_b\\
	[[D_a,D_b],D_{c}]&=0
	\end{align*}
	for all $a,b,c\in\{1,\dots,N\}$.
\end{proof}
\begin{remark}
The theorem above allows us to link the system of $N$ symplectic Dirac operators and their adjoints to the theory of so-called parastatistic algebras from physics in the following way. The parafermion algebra was introduced by Green \cite{green1953generalized} by a system of $N$ parafermion creation and annihilation operators $f_j^{\pm}$ for $j=1,\dots,N$ satisfying \begin{align*}
     [[f_j^{\xi},f_j^{\eta}],f_{\ell}^{\epsilon}] = |\epsilon-\eta|\delta_{k\ell}f_j^{\xi}-|\epsilon-\xi|\delta_{j\ell} f_k^{\eta},
\end{align*}
where $j,k,\ell\in\{1,\dots,N\}$ and $\xi,\epsilon,\eta\in\{+,-\}$. It is well-established that this algebra is in fact isomorphic to the orthogonal Lie algebra $\mathfrak{so}(2N+1)$ (see for instance \cite{ryan1963representations}).  In other words, the operators $D_a$ and $X_a$ for $a\in\{1,\dots,N\}$ can be interpreted as parafermion creation and annihilation operators.
\end{remark}

\noindent The Cartan product from Lemma \ref{cartanLN} now gets the following interpretation: 
\begin{definition}\label{ssm}
	We define the space of \textit{symplectic simplicial monogenics} for $N$ vector variables as follows:
	\begin{align*}
	{\mathcal{S}_{\lambda_1,\dots,\lambda_N}^s(\mathbb{\mathbb{R}}^{2n\times N},\mathbb{S}^{\infty}) := \mathcal{P}_{\lambda_1,\dots,\lambda_N}(\mathbb{\mathbb{R}}^{2n\times N},\mathbb{S}^{\infty}) \cap \ker(D_{s,u_a},\langle \uu_r,\up_{u_s}\rangle, \langle \up_{u_p},\up_{u_q}\rangle_s).  }
	\end{align*}
	for $1\leq r<s\leq N$ and $1\leq p\neq q\leq N$. 
\end{definition}
\section{The symplectic Rarita-Schwinger operator}\label{symplecticRS}
\subsection{Transvector algebras}
Let $\mathfrak{g}$ be a Lie algebra and let $\mathfrak{s}\subset\mathfrak{g}$ be a reductive subalgebra. We then have the decomposition $\mathfrak{g}=\mathfrak{s}\oplus\mathfrak{t}$, where $\mathfrak{t}$ carries an $\mathfrak{s}$-action for the commutator (i.e. $[\mathfrak{s},\mathfrak{t}]\subset\mathfrak{t}$).
 Fix a Cartan subalgebra $\mathfrak{h}$ for $\mathfrak{s}$ and a triangular decomposition $\mathfrak{s}=\mathfrak{s}^-\oplus \mathfrak{h}\oplus \mathfrak{s}^+$, where $\mathfrak{s}^{\pm}$ consists of the positive (resp.\ negative) roots $e_{\alpha}$ (resp. $e_{-\alpha}$) with $\alpha\in\Phi^+(\mathfrak{s})$.
Define a left ideal $J \subset \mathcal{U}(\mathfrak{g})$ in the universal enveloping algebra $\mathcal{U}(\mathfrak{g})$ by means of $\mathcal{U}(\mathfrak{g})\mathfrak{s}^+$ and further define the \textit{normaliser} as $$\operatorname{Norm}(J):=\{u\in \mathcal{U}(\mathfrak{g})\mid Ju\subset J\}.$$ We now have that $J$ is a two-sided ideal of $\operatorname{Norm}(J)$ so that the quotient algebra $\mathcal{S}(\mathfrak{g},\mathfrak{s})=\operatorname{Norm}(J)/J$ which is known as the \textit{Mickelsson algebra}.\par Let us now consider an extension of $\mathcal{U}(\mathfrak{g})$ to a suitable localisation $\mathcal{U}'(\mathfrak{g})$ given by \[ \mathcal{U}'(\mathfrak{g})= \mathcal{U}(\mathfrak{g})\otimes_{\mathcal{U}(\mathfrak{h})}\operatorname{Frac}(\mathcal{U}(\mathfrak{h}))\ ,  \]
where $\operatorname{Frac}(\mathcal{U}(\mathfrak{h}))$ is the field of fractions in the (universal enveloping algebra of the) Cartan algebra.
 The ideal $J'$ can be introduced for this extension too (in a completely similar way) and the quotient algebra $\mathcal{Z}(\mathfrak{g},\mathfrak{s}):=\operatorname{Norm}(J')/J'$ is the \textit{Mickelsson-Zhelobenko algebra}. These two algebras are naturally identified, since one has \[ \mathcal{Z}(\mathfrak{g},\mathfrak{s})= \mathcal{S}(\mathfrak{g},\mathfrak{s})\otimes_{\mathcal{U}(\mathfrak{h})}\operatorname{Frac}(\mathcal{U}(\mathfrak{h})). \]
We refer to this algebra as a `transvector algebra'.
\begin{definition}
	The \textit{extremal projector} for the Lie algebra $\mathfrak{sl}(2)=\operatorname{Alg}\{X,Y,H\}$ is the idempotent operator $\pi$ (meaning that $\pi^2=\pi$) given by the expression 
	\begin{align}
	\pi : = 1 + \sum_{j=1}^{\infty}\frac{(-1)^j}{j!} \frac{\Gamma(H+2)}{\Gamma(H+2+j)}Y^jX^j,
	\end{align}
	where $$\frac{\Gamma(H+2)}{\Gamma(H+2+j)}=\frac{(H+1)!}{(H+j+1)!(H+j)!}$$ is defined using the gamma function $\Gamma$. Note that we will sometimes write $\pi_{\mathfrak{sl}(2)}$ in order to refer to the underlying Lie algebra $\mathfrak{sl}(2)$.
\end{definition}
\begin{remark}
		This operator satisfies $X\pi = \pi Y = 0$ and can therefore be used a natural object to project on the kernel $\ker(X)$. The operator is defined on an extension $\mathcal{U}'(\mathfrak{sl}(2))$ of the universal enveloping algebra so that formal series containing the operator $H$ in the denominator are well-defined.
\end{remark}
\subsection{Higher spin Dirac operators}
The Rarita-Schwinger operator was first introduced in theoretical physics by W. Rarita and J. Schwinger \cite{rarita1941theory} in 1941 to describe fermionic particles of spin-$\frac{3}{2}$, as a generalisation of the Dirac operator which describes particles of spin-$\frac{1}{2}$. As a matter of fact, this operator is an example of a so-called \textit{higher spin Dirac operator} (HSD, \cite{MR1708920}) given by \[ \mathcal{Q}_{\lambda} :\mathcal{C}^{\infty}(\R^n,\vl)\to \mathcal{C}^{\infty}(\R^n,\vl),  \]
where $\vl$ is an irreducible $\Spin(n)$-representation (we assume that $n$ is odd here to avoid the parity of the spinors $\mS$ and write $m:=\lfloor \frac{n}{2}\rfloor$ for notational convenience) with highest weight $$\lambda=\left(\lambda_1+\frac{1}{2},\dots,\lambda_{m-1}+\frac{1}{2},\frac{1}{2}\right)_{\mathfrak{so}(n)}.$$
These are often called \textit{higher spin representations}. Recall from Theorem \ref{vlsctheorem} that these higher spin representations can be explicitly realised as the vector space of $\mathbb{S}$-valued polynomials in several vector variables intersected with the kernel of some differential operators (the simplicial monogenics).
\begin{itemize}
    \item We note that the Dirac operator $$\up_x:\mathcal{C}^{\infty}(\R^n,\mathbb{S})\to \mathcal{C}^{\infty}(\R^n,\mathbb{S})$$ is the easiest example of the operator $\mathcal{Q}_\lambda$ where we take $\lambda_1=\cdots=\lambda_{m-1}=0$, so that we have $\vl=\mathbb{S}$.
    \item  The next type of operator is obtained when taking the highest weight $\lambda=\left(k+\frac{1}{2},\dots,\frac{1}{2}\right)$, such that $\vl=\mathcal{M}_k(\R^n,\mathbb{S})$ (see Section \ref{refinedmonogenicFD}). The corresponding higher spin Dirac operator is called the \textit{Rarita-Schwinger operator}  
\begin{align*}
\mathcal{R}_k : \mathcal{C}^{\infty}(\R^n,\mathcal{M}_k(\R^{2n},\mathbb{S}))\to \mathcal{C}^{\infty}(\R^n,\mathcal{M}_k(\R^{2n},\mathbb{S})).
\end{align*}
	Note that the classical Rarita-Schwinger operator from theoretical physics corresponds to the operator $\mathcal{R}_1$.

\end{itemize}

\noindent
The monogenic refinement of the FD (see Theorem \ref{monogenicref}) in the vector variable $\uu\in\R^n$ can be used to obtain an explicit expression for the first-order differential operator $\mathcal{R}_k$. The decomposition
\[ H_k(\uu) = {M}_{k}(\uu)  + \uu {M}_{k-1}(\uu). \]
gives rise to the following diagram:
\begin{figure}[h!]
	\centering
	\begin{tikzpicture}
	\matrix (m)
	[
	matrix of math nodes,
	row sep    = 0.2em,
	column sep = 4em
	]
	{
		\mathcal{C}^{\infty}(\R^n,\mathcal{M}_k(\R^{2n},\mathbb{S}))              & \mathcal{C}^{\infty}(\R^n,\mathcal{M}_k(\R^{2n},\mathbb{S})) \\
		\rotatebox[origin=c]{90}{$\cong$} &\rotatebox[origin=c]{90}{$\cong$} &\\
		\mathcal{C}^{\infty}(\R^n,\mathcal{M}_k) &    \mathcal{C}^{\infty}(\R^n,\mathcal{M}_k)         \\
		\oplus &	\oplus \\
		\mathcal{C}^{\infty}(\R^n,\underline{u}\mathcal{M}_{k-1}) &    \mathcal{C}^{\infty}(\R^n,\underline{u}\mathcal{M}_{k-1})    \\
	};
	\path
	(m-3-1.east |- m-3-2)
	edge [->] node [above] {$\mathcal{R}_k$} (m-3-2)
	(m-1-1.east |- m-1-2)
	edge [->] node [above] {$\up_x$} (m-1-2);
	\end{tikzpicture}
	\caption{The Rarita-Schwinger operator arising from the Fischer decomposition.}
	\label{RS}
\end{figure}
	Note that there are other arrows that one could study in this scheme (giving rise to e.g.\ the twistor operator) \cite{Bures}. However, we will not consider these here.
\begin{remark}
	Every function $f\in \mathcal{C}^{\infty}(\R^n,\mathcal{M}_k)$ with values in $\mathcal{M}_k$ can be thought of as a function depending on two vector variables $f(\ux;\uu)$ such that fixed $\ux\in\R^n$ the function $f(\ux;\uu)=f_{\ux}(\uu)\in\mathcal{M}_k(\R^n,\mathbb{S})$, i.e.\ $f(\ux;\uu)\in\ker\up_u$. One often calls the variable $\uu$ a \textit{dummy variable}. This illustrates the importance of the polynomial models for $\Spin(n)$ (see Theorem \ref{vlsctheorem}).
\end{remark}
\noindent We summarise these results as follows:
\begin{theorem}[Bure\v{s} et.\ al., \cite{Bures}]
	Let $f(\ux;\uu)\in \mathcal{C}^{\infty}(\R^n,\mathcal{M}_k)$. The Rarita-Schwinger operator is the unique $\Spin(n)$-invariant first-order differential operator \begin{align*}
	\mathcal{R}_k : \mathcal{C}^{\infty}(\R^n,\mathcal{M}_k(\R^{2n},\mathbb{S}))&\to \mathcal{C}^{\infty}(\R^n,\mathcal{M}_k(\R^{2n},\mathbb{S})) \\	f(\ux;\underline{u})&\mapsto \left(1+\frac{\uu \ \up_u}{n+2k-2}\right)\up_xf.
	\end{align*}
	
\end{theorem}
\begin{remark}
	Note that the more general higher spin Dirac operators can be defined using a transvector algebra (see the work of Eelbode \& Raeymaekers \cite{Eelbode2015}). We will follow a similar approach in the symplectic setting instead of the more ad hoc approach in Figure \ref{RS}.
\end{remark}
\subsection{Higher metaplectic Dirac operators}
\noindent Inspired by the definition of higher spin Dirac operators $\mathcal{Q}_{\lambda}$, we now investigate operators $$\mathcal{Q}_k^s :\mathcal{C}^{\infty}(\R^{2n},\vl)\to \mathcal{C}^{\infty}(\R^{2n},\vl),$$
 with values in the (infinite-dimensional) $\mathsf{Mp}(2n)$-representations $\vl$
where $$\lambda=\left(\lambda_1-\frac{1}{2},\dots,\lambda_n-\frac{1}{2}\right)_{\mathfrak{sp}(2n)}\oplus \left(\lambda_1-\frac{1}{2},\dots,\lambda_n-\frac{3}{2}\right)_{\mathfrak{sp}(2n)}. $$
We refer to the operator $\mathcal{Q}_k^s$ as a \textit{higher metaplectic Dirac operator} (or HMD in short). To our knowledge, these operators are not yet studied in the literature.
\begin{example}
	By taking $\lambda_1=0=\cdots=\lambda_n=0$ we have $\vl=\mathbb{S}^{\infty}$ and we recover the symplectic Dirac operator $D_s:\mathcal{C}^{\infty}(\R^{2n},\mS^{\infty})\to \mathcal{C}^{\infty}(\R^{2n},\mS^{\infty}).$
\end{example}
We now take $\lambda_1=k$ and $\lambda_2=\cdots=\lambda_n=0$. Then the corresponding highest weight is given by
$$\lambda=\left(k-\frac{1}{2},\dots,-\frac{1}{2}\right)_{\mathfrak{sp}(2n)}\oplus \left(k-\frac{1}{2},\dots,-\frac{3}{2}\right)_{\mathfrak{sp}(2n)}. $$
In this case, we know that the space of $k$-homogeneous symplectic monogenics $\mathcal{M}_k^s(\R^{2n},\mS)$ is a model for  $\vl$. We consider the elements of $\mathcal{C}^{\infty}(\R^{2n},\mathcal{M}_k^s)$ as functions in two vector variables $f(\ux;\uu)$ which will allow us to use the results obtained in Section \ref{DsxDsu}, in which we defined the symplectic Dirac operators $D_{s,x}$ and $D_{s,u}$. Recall from Lemma \ref{so5} that the Lie algebra generated by the symplectic Dirac operators $D_{s,x}$ and $D_{s,u}$ and their adjoints $X_{s,x}$ and $X_{s,u}$ is isomorphic to $\mathfrak{so}(5)$.  Straightforward computation of commutators leads to:
\begin{lemma}
	The Lie algebra $\mathfrak{sl}(2)=\Alg(D_{s,u},X_{s,u})$ is reductive in $$\mathfrak{so}(5)=\Alg(D_{s,x},D_{s,u},X_{s,x},X_{s,u}).$$ This means that we can write $\mathfrak{so}(5) = \mathfrak{sl}(2)\oplus \mathfrak{t}$ with $\mathfrak{t}$ a subspace satisfying $[\mathfrak{sl}(2),\mathfrak{t}]\subset\mathfrak{t}$.
\end{lemma}
This means that we can construct the transvector algebra $\mathcal{Z}(\mathfrak{so}(5),\mathfrak{sl}(2))$. Taking into account the suitable rescaling the Lie algebra $\mathfrak{sl}(2)\cong \Alg(X,Y,H)$ can be realised as follows (note that this has no effect on the lemma above or the related commutators):
\[ X = \sqrt{2} D_{s,u},\quad Y = \sqrt{2} X_{s,u},\quad H=-(\mathbb{E}_u+\mathbb{E}_v+n) \]
\noindent Recall that the extremal projector $\pi_{\mathfrak{sl}(2)}$ of the Lie algebra $\mathfrak{sl}(2)$ is formally given by the expression \[ \pi_{\mathfrak{sl}(2)}= 1+\sum_{j=1}^{+\infty} \frac{(-1)^j}{j!}\frac{\Gamma(H+2)}{\Gamma(H+2+j)}Y^j X^j.\]

\begin{lemma}
	The transvector projection of the operator $D_{s,x}$ is given by \[ \pi_{\mathfrak{sl}(2)}D_{s,x}=\left(1-\frac{YX}{(H+2)}\right)D_{s,x}. \]
\end{lemma}
\begin{proof}
	We easily find that $[D_{s,x},D_{s,u}]=-\langle \up_x,\up_u\rangle_s$ and $[\langle \up_x,\up_u\rangle_s,D_{s,x}]=0$ so that infinite sum reduces to the finite sum from the lemma.
\end{proof}
\noindent We now come to the following definition of the symplectic Rarita-Schwinger operator:
\begin{definition}
	We define the \textit{symplectic Rarita-Schwinger} operator $\mathcal{R}_k^s$ as 
	\begin{align*}
	\mathcal{R}_k^s:\mathcal{C}^{\infty}(\R^{2n},\mathcal{M}_k^s)&\to \mathcal{C}^{\infty}(\R^{2n},\mathcal{M}_k^s)\\ 
	f(\ux,\uu)&\mapsto \left( 1+\frac{2X_{s,u}D_{s,u}}{k+n+2}\right)D_{s,x} f(\ux,\uu).
	\end{align*}
\end{definition}
\begin{remark}
	Note that the name `symplectic Rarita-Schwinger operator' already appears in the literature; see for example the work of Kr\'ysl  \cite{krysl2007relation} and the references therein. However, their symplectic Rarita-Schwinger operator is defined on symplectic spinor-valued \textit{forms} whereas our operator acts on symplectic spinor-valued \textit{polynomials}. 
\end{remark}
\section{Conclusion and further research}
\noindent In this paper we introduced the notion of symplectic simplicial harmonics and monogenics and related them to the irreducible representations of the symplectic Lie algebra. These results could also be useful for constructing solutions for the symplectic Dirac operator. There are two useful mechanisms in classical Clifford analysis to construct monogenic functions: first of all, there is the Fueter theorem, and secondly the Cauchy-Kowalewskaya extension (CK-extension in short). In the symplectic framework, the Fueter theorem was proven in \cite{fueter}. 
These new polynomial models might be useful to get a grasp on the \textit{symplectic CK-extension} (which has not been proven yet).  Let us focus on the orthogonal case to illustrate what we mean. The association between the irreducible representation $\V_{\lambda}$ and the polynomial model as a kernel space of certain differential operators corresponds with the first dotted horizontal line in Figure \ref{ckbr}.
\begin{figure}[h!]
	\centering
	\begin{tikzpicture}[x=0.75pt,y=0.75pt,yscale=-1,xscale=1]
	
	\draw  [dash pattern={on 0.84pt off 2.51pt}]  (245,110.5) -- (379.5,110.5) ;
	\draw [shift={(381.5,110.5)}, rotate = 180] [color={rgb, 255:red, 0; green, 0; blue, 0 }  ][line width=0.75]    (10.93,-3.29) .. controls (6.95,-1.4) and (3.31,-0.3) .. (0,0) .. controls (3.31,0.3) and (6.95,1.4) .. (10.93,3.29)   ;
	
	\draw    (241,122.5) -- (240.51,210) ;
	\draw [shift={(240.5,212)}, rotate = 270.32] [color={rgb, 255:red, 0; green, 0; blue, 0 }  ][line width=0.75]    (10.93,-3.29) .. controls (6.95,-1.4) and (3.31,-0.3) .. (0,0) .. controls (3.31,0.3) and (6.95,1.4) .. (10.93,3.29)   ;
	
	\draw  [dash pattern={on 0.84pt off 2.51pt}]  (267,229) -- (379,229.49) ;
	\draw [shift={(381,229.5)}, rotate = 180.25] [color={rgb, 255:red, 0; green, 0; blue, 0 }  ][line width=0.75]    (10.93,-3.29) .. controls (6.95,-1.4) and (3.31,-0.3) .. (0,0) .. controls (3.31,0.3) and (6.95,1.4) .. (10.93,3.29)   ;
	
	\draw [color={rgb, 255:red, 0; green, 0; blue, 0 }  ,draw opacity=1 ]   (429,213) -- (429.49,125.5) ;
	\draw [shift={(429.5,123.5)}, rotate = 450.32] [color={rgb, 255:red, 0; green, 0; blue, 0 }  ,draw opacity=1 ][line width=0.75]    (10.93,-3.29) .. controls (6.95,-1.4) and (3.31,-0.3) .. (0,0) .. controls (3.31,0.3) and (6.95,1.4) .. (10.93,3.29)   ;

	\draw (448.5,111) node  [align=left] {polynomial model};
	\draw (231,105) node   {$\mathbb{V}_{\lambda }$};
	\draw (238,226) node   {$\bigoplus _{\mu }\overline{\mathbb{V}}_{\mu }$};
	\draw (238,250.5) node  [align=left] {(lower dimension)};
	\draw (449.5,228) node  [align=left] {polynomial model};
	\draw (254.67,125.67) node   {$n$};
	\draw (266,193.67) node   {$n-1$};
	\draw (197,165) node [color=black  ,opacity=1 ] [align=left] {branching};
	\draw (503,177.5) node [color=black  ,opacity=1 ] [align=left] {CK-extension\\(inverse branching)};

	\end{tikzpicture}
	\caption{Branching, polynomials models and the CK-theorem.}
	\label{ckbr}
\end{figure}
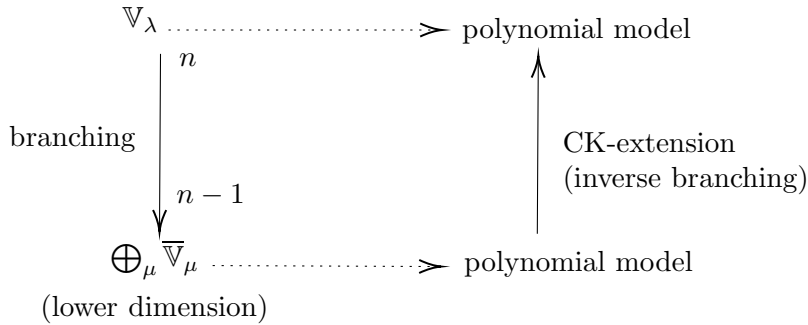

\noindent
The second step, visualised as the left downwards arrow, is to study the behaviour of the $\mathfrak{so}(n)$-irreducible representation $\V_{\lambda}$ as  a representation for the subalgebra $\mathfrak{so}(n-1)$. The representation $\V_{\lambda}$ will ofcourse no longer be irreducible, but will decompose in lower dimensional irreducible representations (which may occur multiple times) for the Lie algebra $\mathfrak{so}(n-1)$. These lower dimensional summands of irreducible representations will also admit polynomial models (this gives rise to the horizontal dotted arrow at the bottom). This means that we are now left with two polynomial models. The problem how these two models interact is then solved by the {CK-extension}. In other words: the CK-extension can be viewed as an \textit{inverse branching} procedure. \par Moreover, one could develop the function theory for the symplectic Rarita-Schwinger operator (and for the more general HMD) like it was done in the orthogonal case (see e.g.\ \cite{DeSchepper2010}).
Last of all, we wish to mention that the connection between Dirac-type operators and parastatistic algebras has recently emerged in \cite{bisbo2022lie}. In this paper, the authors consider a system of $N$ Grassmann Dirac operators and their duals. Recall that the classical Dirac operator $\up_x=\sum_{j=1}^ne_j\partial_{x_j}$ is a contraction between Clifford algebra generators and partial derivatives. When one replaces the derivatives $\partial_{x_j}$ by  Grassmann derivatives (partial derivatives with respect to an anti-commuting or Grassmann variable) $\partial_{\theta_j}$ the resulting Dirac-type operator $\up_\theta=\sum_{j=1}^ne_j\partial_{\theta_j}$ is called the Grassmann Dirac operator and acts on spinor-valued \textit{forms}. The striking difference with the usual Dirac operator $\up_x$ is that $\up_\theta$ and its dual give rise to a copy of the Lie algebra $\mathfrak{sl}(2)$ instead of the Lie super algebra $\mathfrak{osp}(1|2)$, we refer to \cite{slupinski1996hodge} for the details. When considering $N$ Grassmann Dirac operators and their duals, the authors of \cite{bisbo2022lie} show that the algebra closes as $\mathfrak{so}(2N+1)$. Note that we found the same algebra in Theorem \ref{so2k+1} but in terms of symplectic Dirac operators acting on symplectic spinor-valued \textit{polynomials}. This yields a nice application of supersymmetry, in which both the orthogonal and symplectic framework can be combined. This is planned in future work.

\bibliographystyle{abbrv}
\bibliography{biblio}

\begin{thebibliography}{10}

\bibitem{bisbo2022lie}
A.~K. Bisbo, H.~De~Bie, and J.~Van~der Jeugt.
\newblock A {L}ie algebra of {G}rassmannian {D}irac operators and vector
  variables.
\newblock {\em Journal of Lie Theory}, 32(3):751--770, 2022.

\bibitem{BDS}
F.~Brackx, R.~Delanghe, and F.~Sommen.
\newblock {\em Clifford analysis}, volume~76 of {\em Research Notes in
  Mathematics}.
\newblock Pitman (Advanced Publishing Program), Boston, MA, 1982.

\bibitem{brackx2011higher}
F.~Brackx, D.~Eelbode, and L.~{V}an~de {V}oorde.
\newblock Higher spin dirac operators between spaces of simplicial monogenics
  in two vector variables.
\newblock {\em Mathematical Physics, Analysis and Geometry}, 14(1):1--20, 2011.

\bibitem{MR1297597}
D.~J. Britten, J.~Hooper, and F.~W. Lemire.
\newblock Simple {$C_n$} modules with multiplicities {$1$} and applications.
\newblock {\em Canad. J. Phys.}, 72(7-8):326--335, 1994.

\bibitem{bittenlemire}
D.~J. Britten and F.~W. Lemire.
\newblock On modules of bounded multiplicities for the symplectic algebras.
\newblock {\em Trans. Amer. Math. Soc.}, 351(8):3413--3431, 1999.

\bibitem{MR3136522}
D.~Bump.
\newblock {\em Lie groups}, volume 225 of {\em Graduate Texts in Mathematics}.
\newblock Springer, New York, second edition, 2013.

\bibitem{Bures}
J.~Bure\v{s}, F.~Sommen, V.~Sou\v{c}ek, and P.~Van~Lancker.
\newblock {R}arita-{S}chwinger type operators in {C}lifford analysis.
\newblock {\em J. Funct. Anal.}, 185(2):425--455, 2001.

\bibitem{colombo2012analysis}
F.~Colombo, I.~Sabadini, F.~Sommen, and D.~C. Struppa.
\newblock {\em Analysis of Dirac systems and computational algebra}, volume~39.
\newblock Springer Science \& Business Media, 2012.

\bibitem{de2017basic}
H.~De~Bie, M.~Hol{\'\i}kov{\'a}, and P.~Somberg.
\newblock Basic aspects of symplectic {C}lifford analysis for the symplectic
  dirac operator.
\newblock {\em Advances in applied Clifford algebras}, 27:1103--1132, 2017.

\bibitem{MR3126939}
H.~De~Bie, P.~Somberg, and V.~Sou\v{c}ek.
\newblock The metaplectic {H}owe duality and polynomial solutions for the
  symplectic {D}irac operator.
\newblock {\em J. Geom. Phys.}, 75:120--128, 2014.

\bibitem{de2010special}
H.~De~Schepper, D.~Eelbode, and T.~Raeymaekers.
\newblock On a special type of solutions of arbitrary higher spin {D}irac
  operators.
\newblock {\em Journal of Physics A: Mathematical and Theoretical},
  43(32):325208, 2010.

\bibitem{delanghe2012clifford}
R.~Delanghe, F.~Sommen, and V.~Soucek.
\newblock {\em Clifford algebra and spinor-valued functions: a function theory
  for the Dirac operator}, volume~53.
\newblock Springer Science \& Business Media, 2012.

\bibitem{fueter}
D.~Eelbode, S.~Hohloch, and G.~Muarem.
\newblock The symplectic {F}ueter-{S}ce theorem.
\newblock {\em Adv. Appl. Clifford Algebr.}, 30(4):Paper No. 49, 19, 2020.

\bibitem{eelbode2022orthogonal}
D.~Eelbode and G.~Muarem.
\newblock The orthogonal branching problem for symplectic monogenics.
\newblock {\em Advances in Applied Clifford Algebras}, 32(3):32, 2022.

\bibitem{eelbode2023branching}
D.~Eelbode and G.~Muarem.
\newblock Branching symplectic monogenics using a {M}ickelsson--{Z}helobenko
  algebra.
\newblock {\em arXiv preprint arXiv:2301.05066. To appear in Operator Theory:
  Advances and Applications}, 2023.

\bibitem{eelbode2023hermitian}
D.~Eelbode and G.~Muarem.
\newblock A {H}ermitian refinement of symplectic {C}lifford analysis.
\newblock {\em arXiv preprint arXiv:2309.08749}, 2023.

\bibitem{Eelbode2015}
D.~Eelbode and T.~Raeymaekers.
\newblock Polynomial solutions for arbitrary higher spin {D}irac operators.
\newblock {\em Experimental Mathematics}, 24(3):339--354, 2015.

\bibitem{fischer1918differentiationsprozesse}
E.~Fischer.
\newblock {\"U}ber die differentiationsprozesse der algebra.
\newblock {\em Journal f{\"u}r die reine und angewandte Mathematik (Crelles
  Journal)}, 1918(148):1--78, 1918.

\bibitem{frappat2000dictionary}
L.~Frappat, A.~Sciarrino, and P.~Sorba.
\newblock {\em Dictionary on Lie algebras and superalgebras}, volume~10.
\newblock Academic Press San Diego, CA, 2000.

\bibitem{gilbert1991clifford}
J.~E. Gilbert and M.~A.~M. Murray.
\newblock {\em Clifford algebras and Dirac operators in harmonic analysis}.
\newblock Number~26. Cambridge university press, 1991.

\bibitem{Goodman2009}
R.~Goodman and N.~R. Wallach.
\newblock {\em Symmetry, Representations, and Invariants}.
\newblock Springer New York, 2009.

\bibitem{green1953generalized}
H.~S. Green.
\newblock A generalized method of field quantization.
\newblock {\em Physical Review}, 90(2):270, 1953.

\bibitem{habermann2006introduction}
K.~Habermann and L.~Habermann.
\newblock {\em Introduction to symplectic Dirac operators}.
\newblock Springer, 2006.

\bibitem{howe1989remarks}
R.~Howe.
\newblock Remarks on classical invariant theory.
\newblock {\em Transactions of the American Mathematical Society},
  313(2):539--570, 1989.

\bibitem{krysl2007relation}
S.~Kr{\`y}sl.
\newblock Relation of the spectra of symplectic {R}arita-{S}chwinger and
  {D}irac operators on flat symplectic manifolds.
\newblock {\em Archivum Mathematicum}, 43(5):467--484, 2007.

\bibitem{lavivcka2019separation}
R.~L{\'a}vi{\v{c}}ka.
\newblock Separation of variables in the semistable range.
\newblock {\em Topics in Clifford Analysis: Special Volume in Honor of Wolfgang
  Spr{\"o}{\ss}ig}, pages 395--403, 2019.

\bibitem{rarita1941theory}
W.~Rarita and J.~Schwinger.
\newblock On a theory of particles with half-integral spin.
\newblock {\em Physical Review}, 60(1):61, 1941.

\bibitem{ryan1963representations}
C.~Ryan and E.~Sudarshan.
\newblock Representations of parafermi rings.
\newblock {\em Nuclear Physics}, 47:207--211, 1963.

\bibitem{DeSchepper2010}
H.~D. Schepper, D.~Eelbode, and T.~Raeymaekers.
\newblock On a special type of solutions of arbitrary higher spin {D}irac
  operators.
\newblock {\em Journal of Physics A: Mathematical and Theoretical},
  43(32):325208, July 2010.

\bibitem{slupinski1996hodge}
M.~Slupinski.
\newblock A {H}odge type decomposition for spinor valued forms.
\newblock In {\em Annales scientifiques de l'Ecole normale sup{\'e}rieure},
  volume~29, pages 23--48, 1996.

\bibitem{van2012monogenic}
P.~Van~Lancker.
\newblock The monogenic fischer decomposition: two vector variables.
\newblock {\em Complex Analysis and Operator Theory}, 6:425--446, 2012.

\bibitem{MR2106725}
P.~Van~Lancker, F.~Sommen, and D.~Constales.
\newblock Models for irreducible representations of {${\rm Spin}(m)$}.
\newblock {\em Adv. Appl. Clifford Algebras}, 11(S1):271--289, 2001.

\end{thebibliography}
\end{document}